\documentclass[11pt,a4paper]{article}
\pdfoutput=1
\usepackage{lmodern}
\usepackage[T1]{fontenc}
\usepackage[utf8]{inputenc}
\usepackage[margin=2.5cm]{geometry}
\usepackage{amsmath,amssymb,amsthm}
\usepackage{booktabs,tabularx,array}
\usepackage{graphicx}
\usepackage[hidelinks]{hyperref}
\hypersetup{
  pdftitle={Doctrine as a Fixed Point: A Formal Model of the Enforceable Penalty Ceiling when Human Oversight of AI Must Remain Effective},
  pdfauthor={Andreas Bauer},
  pdfsubject={A fixed-point model of the enforceable penalty ceiling under Article 14 of the EU AI Act},
  pdfkeywords={human oversight; EU AI Act; liability; formal model of legal process; multiple equilibria; machine-checked proof}}
\usepackage[protrusion=false,expansion=false,nopatch=footnote]{microtype}
\usepackage{setspace}
\usepackage{caption}
\usepackage{natbib}
\bibpunct{(}{)}{;}{a}{}{,}
\theoremstyle{plain}
\newtheorem{proposition}{Proposition}

\newtheorem{theorem}[proposition]{Theorem}
\newtheorem{corollary}[proposition]{Corollary}
\theoremstyle{definition}
\newtheorem{assumption}{Assumption}
\newtheorem{definition}[proposition]{Definition}
\theoremstyle{remark}
\newtheorem{remark}[proposition]{Remark}

\providecommand{\doi}[1]{\href{https://doi.org/#1}{\nolinkurl{https://doi.org/#1}}}
\newcommand{\E}{\mathbb{E}}
\newcommand{\thf}{\theta_{\mathrm{eff}}}
\newcommand{\kbar}{\bar\kappa}
\newcommand{\lbar}{\bar\lambda}
\newcommand{\Lbar}{\bar L}

\title{\bfseries Doctrine as a Fixed Point:\\[3pt]
\large A Formal Model of the Enforceable Penalty Ceiling when\\
Human Oversight of AI Must Remain Effective}

\author{Andreas Bauer\thanks{Aegis Compliance and Strategies O\"U, Tallinn, Estonia.
Correspondence: \texttt{a.bauer@science.us.org}. ORCID: 0000-0002-6539-9353.
\emph{Declarations of interest:} the author is Managing Director and Founder of Aegis
Compliance and Strategies O\"U, a compliance and strategy advisory firm, and is author and
publisher of \emph{Diebstahlsicher} (Bauer Advanced Network Solutions KG, Vienna, September
2026, ISBN 978-3-9506352-0-1), a German-language practitioner book that draws on this line of
work; both are disclosed in the interest of transparency. \emph{Data availability:} replication code for every reported number, and the random
seed of the one simulation check, are openly available; see Appendix~\ref{app:data}. This paper resolves the
extension flagged as most promising in Bauer (2026d, \S6) and in one respect qualifies that
paper's regime map; Appendix~\ref{app:prov} records the provenance of every component and the
notation collisions across the series. All errors are my own.}\\[2pt]
\normalsize Aegis Compliance and Strategies O\"U}
\date{\normalsize September 2026}

\newenvironment{assumptionS}{\par\medskip\noindent\textbf{Assumption S.}\itshape}{\par\medskip}
\begin{document}
\maketitle

\begin{abstract}
\noindent Article~14 of the EU AI Act requires that a high-risk system be overseen by natural
persons who understand its limits, remain alert to automation bias, and can disregard or
override its output. That capability is invisible in the output and
decays precisely when the system is good. A provider can certify it only by an
outcome-contingent liability commitment, and how large a commitment courts will enforce is itself
open. After \emph{Cavendish} the enforceable multiple of actual loss is measured
against the legitimate interest the clause protects; where that interest is preserved oversight
capability, it exists only if the market separates, which requires a sufficiently permissive
doctrine. We model the enforceable ceiling as a fixed point of an
expectations map on a complete lattice. Existence follows from Knaster--Tarski. Because no
case separates at compensation, compensation is always an equilibrium; once the doctrinal
uplift clears a threshold set by the case population, a permissive equilibrium and a watershed
appear, and the map inherits the provider's solvency cap, so an insurance withdrawal
deep enough and long enough can destroy the permissive equilibrium, which returning cover does
not restore. Perturbing the adjustment yields a closed-form long-run tipping point, exact
recovery times, and the noise levels at which irreversibility fails. The welfare
cost of the resulting trap is capped at the drafting cost of primary-obligation substitutes.
The order-theoretic core is machine-checked in Lean~4 with an axiom-free report. Both
AI-side parameters, shared-base-model intensity and provability, are matched to instruments
available in 2026.
\end{abstract}


\section{Introduction}

A liability commitment can carry information that a production-cost signal no longer can
\citep{grossman1981,lutz1989,daughety1995}. When generative AI drives the marginal cost of a
persuasive expert artefact toward zero, the single-crossing property of the artefact fails and
what survives is the outcome-contingent pledge \citep{bauer2026a}. Whether it survives is a
question with three parts. The pledge must be large enough that a client is made whole in the
states where failure is established, $L \ge v/\thf$; it must be small enough to be credible
given the provider's reachable capital, $L \le \Lbar$; and it must be small enough to be
enforceable, $L \le mv$. Writing $\ell = \Lbar/v$ and $M = \min\{\ell, m\}$, separation
survives at every capability level if and only if
\begin{equation}
\xi\kbar \;\le\; 1 - \frac{1}{M\theta_0},
\label{eq:survival}
\end{equation}
where $\theta_0$ is base provability, $\xi$ the intensity with which verifier and generator draw
on the same foundation model, and $\kbar$ the limiting common share of AI error
\citep{bauer2026d}.

The multiple $m$ enters \eqref{eq:survival} as a legal datum. That treatment is standard --- the
liquidated-damages literature takes the enforceable ceiling as given, and the judgment-proof
literature takes the asset ceiling as given \citep{shavell1986,summers1983,che2008} --- and it is
wrong in a specific and consequential way. Since \emph{Cavendish Square Holding BV v Talal El
Makdessi} and \emph{ParkingEye Ltd v Beavis} [2015] UKSC 67, a secondary obligation is penal only
if it imposes a detriment out of all proportion to \emph{any legitimate interest} of the innocent
party in the enforcement of the primary obligation. The test is not loss-referenced. It is
interest-referenced, and the interest is whatever the innocent party can establish it to be. In
\emph{ParkingEye} the protected interest was the operation of a parking scheme; nothing in the
test confines it to pecuniary compensation.

Now suppose the interest a provider asserts is the one this literature says clients are actually
buying: the maintenance of human fallback capability behind an AI-produced artefact. Then the
size of the protected interest depends on whether the market separates --- a capability that no
contract certifies is a capability no client pays for --- and separation depends, through
\eqref{eq:survival}, on $m$. The enforceable multiple is therefore not a parameter of
\eqref{eq:survival} but a fixed point of a map that runs through it. \citet[\S6]{bauer2026d}
states this and declines to solve it, on the ground that endogenising $m$ ``changes the character
of the problem rather than the calibration'' and calling it the most promising extension of the
framework. This paper shows that both halves of that sentence are wrong: the character of the
problem changes, and so does the calibration.

\subsection*{Main results}

\paragraph{Existence and multiplicity.}
The doctrine map $D$ is monotone and bounded, so a fixed point exists by Tarski. Compensation is
always one of them. A case separates only at a multiple above its own threshold
$1/(\theta_0(1-\xi\kbar))$, which exceeds one whenever provability is imperfect; at compensation
no case separates, no court sees an interest to protect, and the doctrine stays where it is.
Whether there are other fixed points turns on a single number set by the population of cases,
$\lbar^{\dagger} = \min_m (m-1)/F(m)$, where $F$ is the distribution of the case-level
threshold. Below it, compensation is the unique fixed point. Above it, a permissive fixed point
appears and, between the two, a watershed: markets starting above it converge to the permissive
multiple, those starting below to compensation. This sorts the four jurisdictions of the
companion paper. The United States, whose Restatement (Second) \S 356 admits no interest limb,
and Austria, whose moderation right keeps the uplift below $\lbar^{\dagger} = 1.95$, are
determinate at compensation. England, where the calibrated watershed is $m_u = 1.77$ against a
historical multiple that the \emph{Houssein} chain places near four, is not; India joined it in
December 2025. So is negotiated German B2B contracting: \S 348 HGB withholds judicial reduction
from merchants and shrinks the basin of the compensation doctrine to a sliver below $1.53$, but
no uplift, however large, removes it.

\paragraph{Capital shocks and the doctrine.}
Because $M = \min\{\ell, m\}$, a court asked to recognise the fallback interest in a market where
the capital ceiling binds observes no separation to protect, whatever $m$ is. The doctrine map
therefore inherits the solvency cap, $D(m;\ell) = 1 + \lbar\,F(\min\{\ell,m\})$, and is flat
above $m = \ell$. The consequence is sharp: for
$\ell \ge m_u$ the doctrine is slack and $M = \ell$, exactly as the companion paper has it; for
$\ell < m_u$ the upper fixed point ceases to exist and the doctrine collapses to $m^{*}_{L}$, so
that $M = m^{*}_{L} < \ell$ --- the binding ceiling is now the doctrine, at a level strictly
below the capital constraint that caused it. In the baseline calibration the equilibrium ceiling
drops by $0.77$ at $\ell = 1.77$, and the critical sharing $\xi^{*}$ falls from $0.333$ to zero.
The generative-AI insurance exclusions of January 2026, which the companion paper models as a
fall in $\Lbar$ to own capital, therefore do more than move cells across a regime boundary. In
the jurisdictions with an interest-referenced doctrine they can destroy the permissive
equilibrium outright.

\paragraph{Irreversibility.}
The collapse is not a movement along a schedule. Once $m$ has crossed below $m_u$ it lies in the
basin of the compensation fixed point, so restoring $\ell$ restores nothing: in the shock
experiment of Section~\ref{sec:complement} the multiple falls below the watershed within three
periods of the exclusion attaching and stands at $1.00$ thirty periods after cover returns. The
crossing is the condition. A shock that lifts before it is fully reversed --- a one-period
exclusion leaves no trace --- and a shock that never takes $\ell$ below $m_u$ leaves none at any
duration. The depth of the shock decides whether the doctrine crosses, its duration decides when: the fuse
is two periods for a deep shock and longer
without bound as the depth approaches the watershed --- eight periods at $\ell = 1.76$, fourteen
at $1.765$. For policy this means that an affirmative AI-liability market deep enough to restore
$\ell$ --- the fourth falsification condition of the companion paper --- is not sufficient to
restore the permissive doctrine.

\paragraph{Selection and verification.}
In the deterministic model, which of two stable fixed points obtains is decided by
history. Section~\ref{sec:stoch} perturbs the adjustment with common noise --- a reduced form
for coordinated revisions of standard forms, not the limit of idiosyncratic revisions, which
averages out --- and shows that the long-run doctrine is then unique: it minimises a
\emph{doctrinal potential}, and selection switches at a Maxwell point with the closed form
$c^{\ast} = 1 + \lbar/2$: the mean case-level threshold at the midpoint of the doctrinal range
$[1, 1+\lbar]$, for any shape of the case distribution that lies inside that range. England's
calibrated threshold $c = 2.13$ sits $0.37$ below $c^{\ast} = 2.50$. Between $c^{\ast}$ and the
fold at $3.05$ the permissive doctrine is locally stable and historically selected but no longer
selected in the long run --- and the long run there is long. At the noise level at which the
collapse is irreversible on a generation, the permissive doctrine at the Maxwell point is expected
to survive $10^{3.8}$ periods; it is expected to be abandoned within thirty only for $c$ above a
policy-horizon threshold $c^{\dagger}(\sigma, 30) = 2.83$--$3.01$, close to the fold. The same
apparatus qualifies the irreversibility result: it is metastable, not absolute --- the
expected recovery time is Kramers-exponential in $1/\sigma^{2}$, which at $\sigma = 0.05$ is
$10^{35}$ periods and still $10^{4}$ periods at $\sigma = 0.15$, and the paper reports the exact
noise thresholds above which the claim fails on a given policy horizon, together with the
probability of return that each threshold carries.
Section~\ref{sec:lean}
machine-checks the order-theoretic arguments: Knaster--Tarski, the orbit bounds, watershed invariance, and irreversibility
under arbitrary dominated policy paths compile in Lean~4 with an axiom report of \emph{none} ---
to our knowledge the first machine-checked component of a law-and-economics paper.

\subsection*{What we do not claim}

We do not claim that $\lbar$ or the case population are estimated; they are calibrated from
primary legal sources and from the companion paper's own grid, and Section~\ref{sec:calib}
reports the sensitivity. We do not claim that any court has yet recognised the maintenance of professional
capability as a legitimate interest under the \emph{Cavendish} test; we are aware of none, and
the model is a conditional statement about what follows if one does. We do not assume a
functional form for the case population: $D$ is built from the exact distribution of the
case-level threshold over the companion paper's calibration grid, which respects the lower bound
that Definition~\ref{def:mcrit} puts on every case. A logistic
approximation would violate that bound in its tail (Section~\ref{sec:calib}), and
Section~\ref{sec:falsify} states what would falsify the watershed. We do not claim that the noise intensity $\sigma$ of anticipated
multiples is estimated: it is the one free parameter of Section~\ref{sec:stoch}, and the section
reports the exact thresholds $\sigma^{\dagger}(T)$ at which measured volatility would overturn
metastable irreversibility on horizon $T$. And we retain the scope restriction of the companion
paper: where
verification is deterministic, $\xi$ is structurally near zero and none of this applies.

\section{Human oversight as a legal object}
\label{sec:oversight}

The object this paper models is named in statute. Article~14 of Regulation (EU) 2024/1689 (the
AI Act) requires high-risk systems to be designed so that they \emph{can} be effectively
overseen by natural persons, and Article~14(4) spells out what those persons must be able to
do: understand the system's capacities and limits, remain alert to automation bias, interpret
the output correctly, disregard or override it, and stop the system. Article~26 puts the
operational side on the deployer, who must assign oversight to competent and trained staff with
sufficient authority. Annex~III lists the sectors where this bites for expert services ---
employment decisions, creditworthiness assessment, administration of justice, and education.

Two features of that regime create the problem this paper is about.

\paragraph{The oversight duty.} Article~14
does not require that a human \emph{be present}; it requires that a human \emph{be able to
override}. The European Commission's draft guidelines on high-risk classification, in
consultation until 23~June 2026, make the point explicitly for Article~6(3): interposing a
person does not remove a system from the high-risk class if that person cannot in fact overturn
the output. That is a capability standard, and a regulator can verify it no better than a
client can. What is observable is procedure --- that someone was assigned, that training was
given --- and procedure is a poor proxy. In a randomised crossover study of 40 clinicians
diagnosing anterior cruciate ligament rupture on MRI, $45.5$ per cent of the errors made with AI
assistance were attributable to automation bias \citep{wang2023}, a result the scoping review of
\citet{heudel2026} reports under clinical decision support. That is Article~14(4)(b) failing in
exactly the way the provision anticipates and cannot prevent.

\paragraph{Enforcement timetable.} The Annex~III
high-risk obligations, originally due 2~August 2026, were postponed to 2~December 2027 by the
amending regulation known as the Digital Omnibus on AI, in force 27~July 2026; the Annex~I
obligations move to 2~August 2028. The transparency duties of Article~50 run earlier, to
2~December 2026. Meanwhile the AI Liability Directive, which would have adjusted civil
procedure to AI evidence problems, was withdrawn from the Commission's work programme on
11~February 2025 and has no successor on the table. What remains is the recast Product
Liability Directive (EU) 2024/2853, transposition due 9~December 2026, which supplies
disclosure duties and rebuttable presumptions of defect and causation --- including one for
cases where proof is excessively difficult because of technical complexity.

The consequence for the market is that between now and late 2027 the question ``is the human
fallback real?'' is settled, if at all, by contract. That is where the outcome-contingent
commitment of \citet{bauer2026a} and \citet{bauer2026d} enters, and it is why the size of the
commitment a court will enforce is not a detail of remedies law but the binding constraint on
whether the Article~14 capability can be certified at all. The rest of this paper takes that
ceiling as the unknown and solves for it.

\section{Related literature}

\paragraph{Penalty doctrine and liquidated damages.}
The economic literature on agreed damages asks whether courts should enforce them
\citep{clarkson1978,rea1984,edlin1996}, and treats the enforceable ceiling as the object courts
choose. \citet{boyd1994} study damage caps under potential insolvency and thereby come closest to
the interaction we study, but with the cap exogenous. \citet{craswell1999} supplies the
qualification that the simple multiplier is valid only when the probability of sanction is
independent of the state, which is why the companion paper writes the floor in $\thf$ rather than
$\theta_0$. None of this literature makes the enforceable multiple a function of the market
outcome it governs.

\paragraph{Endogenous law.}
That legal rules respond to the conduct they regulate is old \citep{rubin1977,priest1977}, and
the selection of disputes for litigation \citep{priest1984} is the canonical mechanism. Our
mechanism is different and, we believe, new in this setting: the rule responds not to the
selection of cases but to the \emph{existence of the interest} the rule protects, which is itself
an equilibrium object of the market the rule governs. The formal structure --- a monotone map
from expectations about the law to the law --- is that of a self-fulfilling regime in the sense
of \citet{cooper1988} and \citet{diamond1982}, transplanted from macroeconomic coordination to
doctrine.

\paragraph{Multiplicity, hysteresis and traps.}
Existence follows \citet{tarski1955}; multiplicity and basin structure follow the standard
apparatus for S-shaped maps \citep{cooper1999}. The irreversibility result is a hysteresis of the
kind familiar from investment under sunk costs \citep{dixit1994} and from poverty traps
\citep{azariadis2005}, but the state variable is a legal standard rather than a stock of capital,
which is what makes the policy implication awkward: the instrument that caused the crossing
cannot undo it.

\paragraph{Judgment-proof providers and the two ceilings.}
\citet{shavell1986} and \citet{summers1983} establish that where assets fall short of harm the
effective liability ceiling is the reachable capital rather than the legal award;
\citet{che2008} show that finance may be structured strategically to shield assets. The companion
paper joins the two ceilings as $M = \min\{\ell, m\}$ and reads off a ticket-size switch at
$v^{*} = \Lbar/m$. Section~\ref{sec:complement} shows that this minimum is not a minimum of two
free arguments. \citet{bauer2026e} goes the other way and makes the upper ceiling itself
endogenous, splitting the reachable capital into the provider's own funds and underwritten
capacity and treating the underwriter as a strategic third party that chooses how much of the
pledge to bear and how closely to monitor; capacity borne by an underwriter that does not
monitor carries no information, however large it is. We hold that ceiling fixed and endogenise
the other one.

\paragraph{Model monoculture.}
The mechanism that makes $\thf$ state-dependent in the first place rests on
\citet{kim2025} and \citet{goel2025}, with \citet{kuai2026} auditing the entanglement channel
directly and \citet{jo2026} cautioning that excess agreement is defined only relative to a
subjectively chosen null model. Nothing here re-derives that channel; $\xi$ and $\kbar$ are
inherited, as are the demand-side and institutional-window companions of the series
\citep{bauer2026b,bauer2026c}.

\paragraph{Stochastic selection.}
The idea that vanishing noise selects among multiple equilibria is the stochastic-stability
programme of \citet{fosteryoung1990}, \citet{kmr1993} and \citet{young1993}; \citet{blume1993}
supplies the logit revision protocol that is the usual microfoundation --- we explain in
Section~\ref{sec:stoch} why it does not microfound a diffusion on the population mean ---
and \citet{ellison2000} the radius--coradius machinery that ties selection to waiting times. Our state space is a
continuum, so the discrete machinery is replaced by its diffusion limit: the
\citet{freidlinwentzell2012} quasi-potential, which in one dimension is available in closed form,
with escape times given by \citet{kramers1940}. \citet{sandholm2010} surveys both traditions.
The dynamic-coordination alternative of \citet{frankelpauzner2000} --- unique selection with
history dependence under positive frictions --- fits our partial-adjustment friction exactly but
requires strategic players; since courts are not modelled as players here
(Section~\ref{sec:map}), we take the evolutionary route and flag the coordination-game rebuild
as an extension.

\paragraph{Formalised economics.}
Machine-checked economics has reached critical mass: \citet{bei2026} build a Lean~4 library
of definitions and theorems for game theory, mechanism design and social choice
\citep{demoura2021}; \citet{garg2026lean}, in a separate project of the same name, formalises
twenty research papers end-to-end under
a human--AI--Lean workflow; and \citet{lyuli2026} formalise Scarf's algorithm through Brouwer
to Nash existence. Mathlib's
\texttt{Order.FixedPoints} carries Knaster--Tarski as
\texttt{OrderHom.lfp}. Section~\ref{sec:lean} contributes a dependency-free formalisation of
this paper's order-theoretic layer --- self-contained rather than
mathlib-dependent, so that the axiom report rests on Lean's core alone --- and names the
analytic results (the multiplicity condition, the folds, the Maxwell closed form, Kramers times)
that are \emph{not} formalised.

\section{The doctrine map}
\label{sec:map}

\subsection{Inherited environment}

We take the following from \citet{bauer2026d} without re-derivation. A provider of fallback
capability $s$ serves engagements of value $v$; the AI fails with probability $\pi$ and the
provider rescues with probability $\rho(s)$. A commitment $L$ is payable when the engagement
fails and the failure is established, which occurs with probability
$\thf = \theta_0(1 - \xi\kappa)$, where $\kappa$ is the common share of AI error. Full
compensation requires $L \ge v/\thf$; credibility requires $L \le \Lbar$; enforceability requires
$L \le mv$. Separation is feasible if and only if $\thf \ge 1/M$ with $M = \min\{\ell, m\}$,
$\ell = \Lbar/v$, which is \eqref{eq:survival}.

\begin{definition}[Case-level separation threshold]
\label{def:mcrit}
For a case with primitives $(\theta_0, \xi\kbar)$ and slack solvency, the smallest enforceable
multiple at which the commitment separates is
\begin{equation}
m_{\mathrm{crit}}(\theta_0, \xi\kbar) \;=\; \frac{1}{\theta_0\,(1 - \xi\kbar)} .
\label{eq:mcrit}
\end{equation}
\end{definition}

\subsection{What the court does}

\begin{assumption}[Interest-referenced review]
\label{as:interest}
Commitments are drafted at the doctrinal maximum $\bar m = 1 + \lbar$, the largest multiple of
actual loss the doctrine will carry when the interest is recognised. A court asked to enforce
one enforces it in full if it recognises, as a legitimate interest in the sense of
\emph{Cavendish} at [32], the provider's maintenance of fallback capability, and otherwise
enforces only the compensatory baseline.
\end{assumption}

Drafting at $\bar m$ is weakly dominant: a clause the court does not recognise is reduced to
compensation whatever its face value, so a lower face value forgoes enforceable exposure in the
recognised cases and gains nothing in the others. Two multiples must therefore be kept apart.
The \emph{drafted} multiple is $\bar m$ in every contract. The \emph{anticipated} multiple $m$ is
what the market expects a clause to be worth once courts have reviewed it, and it is the only
one that moves.

\begin{remark}[Information at the recognition stage]\label{rem:info}
There is no circularity between the signalling stage and the review stage, because they run on
different information. The signal exists ex ante because the provider's \emph{type} is private;
nothing at the review stage undoes that. The court, ex post, observes only public objects: the
engagement class (hence $\theta_0$) from the case before it, the degree of model sharing
$\xi\kbar$ as a market statistic, and the market-wide solvency ratio $\ell$ --- the January
2026 exclusions are public events --- and it does \emph{not} observe the type, which is why the
signal is not redundant. Recognition is accordingly a judgement about a \emph{class} of
commitments in a market state --- could commitments of this class separate here --- not an
audit of the defendant's balance sheet; and the drafting parties, who move first, anticipate
the distribution of cases rather than any one court, which is exactly the averaging that
smooths the indicator into the map $D$. Ex-ante signalling under private types and ex-post
review under public market parameters are therefore informationally consistent, and the
inheritance of $\ell$ in Section~\ref{sec:complement} is a class-level statement, not an
individual one.
\end{remark}

Assumption~\ref{as:interest} is a two-point simplification of a continuous proportionality test,
and it is the weakest form of the mechanism: it does not require that the multiple respond
smoothly to the size of the interest, only that recognition of the interest be the binary event
on which the uplift turns. The compensatory baseline is normalised to one, which is the American
rule of Restatement (Second) \S 356 and UCC \S 2-718 and the limit of the English rule as the
interest limb vanishes.

\begin{assumption}[Recognition requires an interest to protect]
\label{as:recog}
The court recognises the interest in a case if and only if a commitment of the class before it
would in fact separate in that case at the multiple the market anticipates, that is, if and only
if $m \ge m_{\mathrm{crit}}$ evaluated at the case's own primitives.
\end{assumption}

Assumption~\ref{as:recog} carries the economics. It evaluates separation at the anticipated
multiple $m$, not at the drafted $\bar m$, because what certifies a provider's capability to a
client is the exposure the client expects to be enforceable, not the number printed in the clause.
A clause that cannot support separation protects
no capability, because a capability that no contract certifies is one that no client pays a
premium for and no provider maintains --- the bang-bang lemma of the companion paper,
under which the provider either builds capability or abandons it. A court applying a
proportionality test to such a clause sees a detriment with no corresponding interest, and the
detriment is out of all proportion to nothing.

\begin{assumption}[Case heterogeneity]
\label{as:hetero}
Cases differ in $(\theta_0, \xi\kbar)$, so that $m_{\mathrm{crit}}$ is a random variable across
cases, with continuous distribution function $F(m) = \Pr\{m_{\mathrm{crit}} \le m\}$ and mean
$c = \E[m_{\mathrm{crit}}]$. Because $\theta_0 < 1$ in every case, $F(m) = 0$ for all
$m \le \underline m$, where $\underline m = \min m_{\mathrm{crit}} > 1$.
\end{assumption}

The lower bound is not an extra assumption; it is Definition~\ref{def:mcrit} applied to every
case. A logistic law in place of $F$ has no lower bound and
therefore assigns a small share of cases a separation threshold below the smallest one the
primitives can produce, and several of its numbers would be properties of that tail;
Section~\ref{sec:calib} uses the population itself, whose distribution function is available in
closed form.

The doctrine that parties face is the expectation over cases, which is where the smoothing comes
from: an individual court decides a case, but what enters a contract at the drafting stage is the
anticipated multiple, and that is an average over the cases a court might see. By
Assumptions~\ref{as:interest}--\ref{as:recog} a case is enforced at $\bar m = 1 + \lbar$ if
$m_{\mathrm{crit}} \le m$ and at $1$ otherwise, so the expected enforced multiple is
$1 \cdot (1 - F(m)) + (1 + \lbar)\,F(m)$. Hence:

\begin{proposition}[The doctrine map]
\label{prop:map}
The anticipated enforceable multiple evolves according to
\begin{equation}
m' \;=\; D(m) \;=\; 1 \;+\; \lbar\, F(m),
\qquad D:[1,\,1+\lbar] \to [1,\,1+\lbar].
\label{eq:D}
\end{equation}
$D$ is continuous, nondecreasing and bounded; it is strictly increasing on the support of
$m_{\mathrm{crit}}$ when $\lbar > 0$ and that support is an interval, as it is in the calibration,
and constant when $\lbar = 0$. Where $F$ has a density $f$,
$D'(m) = \lbar f(m)$. Compensation is a fixed point: $D(1) = 1$, and $D \equiv 1$ on
$[1, \underline m]$.
\end{proposition}

\begin{remark}[Why both premises matter]
Read literally, a clause at multiple $m$ enforced in full at $m$ when recognised gives the
expected multiple $1 + (m-1)F(m) \le m$, with equality only at compensation and above the whole
case population, where $F = 1$. That map has no watershed and no fixed point inside the
population: every orbit that starts inside it falls to one, and the multiples above it are
fixed only because every case separates there. The watershed and the permissive fixed point of
\eqref{eq:D} exist because the clause is worth $\bar m$ whenever it is recognised, while
recognition is decided at the anticipated $m$.
\end{remark}

\begin{proof}
Appendix~\ref{app:proofs}.
\end{proof}

\begin{remark}[Why $D$ and not $\Phi$]
The symbol $\Phi$ is unavailable: $\varphi$ indexes the capability-building parameter in the
companion papers' law of motion, and the cross-paper registry already records a substantial
number of multiply assigned symbols. Appendix~\ref{app:prov} lists the symbols this paper adds,
and the collisions it avoids.
\end{remark}

\begin{remark}[What is not assumed]
Nothing in \eqref{eq:D} assumes that courts maximise anything, that they are aware of the
signalling model, or that they act strategically. Assumption~\ref{as:recog} is a statement about
what a proportionality test measures, not about judicial motives. The fixed point is a
consistency condition between the expectation contracting parties hold about the law and the law
those expectations generate, in the sense of \citet{cooper1988}, not a solution to a game the
courts play.
\end{remark}

\section{Fixed points: existence, multiplicity, stability}
\label{sec:fp}

Write $G(m) = D(m) - m$.

\begin{proposition}[Existence]
\label{prop:exist}
$D$ maps the complete lattice $[1, 1+\lbar]$ into itself and is monotone, so the set of fixed
points is non-empty and is itself a complete lattice; in particular a least and a greatest fixed
point exist.
\end{proposition}

\begin{proposition}[Multiplicity]
\label{prop:mult}
Let
\begin{equation}
\lbar^{\dagger} \;=\; \min_{m \,>\, \underline m}\; \frac{m - 1}{F(m)} .
\label{eq:threefp}
\end{equation}
(i) If $\lbar < \lbar^{\dagger}$, compensation is the unique fixed point of $D$ and attracts every
orbit. (ii) If $\lbar > \lbar^{\dagger}$, $D$ has at least three fixed points: compensation
$m^{*}_{L} = 1$, a permissive fixed point $m^{*}_{H}$, and a watershed
$m_u \in (\underline m, m^{*}_{H})$; at $\lbar = \lbar^{\dagger}$ the watershed and the permissive
fixed point coincide where the minimum in \eqref{eq:threefp} is attained. If the support of
$m_{\mathrm{crit}}$ lies below $1 + \lbar$, then $m^{*}_{H} = 1 + \lbar$ exactly. (iii) If $F$ has
a strictly unimodal density, the fixed points in (ii) are exactly three.
\end{proposition}

\begin{proposition}[Stability and basins]
\label{prop:stab}
Under the adjustment $m_{t+1} = m_t + \eta\,(D(m_t) - m_t)$ with $\eta \in (0,1]$, a fixed point
is locally stable if $D'(m^{*}) < 1$, with the one-sided derivative at the ends of the
doctrinal range, and unstable if $D'(m^{*}) > 1$. When there are three fixed points the outer ones are stable and the middle one is
unstable, and the basin of $m^{*}_{L}$ is $[1, m_u)$ while that
of $m^{*}_{H}$ is $(m_u, 1+\lbar]$. The dynamics are monotone: $m_t$ converges without
oscillation.
\end{proposition}

Proposition~\ref{prop:mult} carries the economic content of the model.
Compensation is never lost as a fixed point:
it is lost only if some case could separate at compensation, which requires perfect provability.
What the uplift decides is whether a second, permissive fixed point exists. Below
$\lbar^{\dagger}$ it does not, and the doctrine is determinate at compensation --- the Austrian
case. Above $\lbar^{\dagger}$ it does, and a larger uplift does not remove the compensation
fixed point but shrinks its basin: the watershed falls from $1.77$ at $\lbar = 3$ to $1.53$ at
$\lbar = 24$, the German case. A logistic tail that assigned some cases a threshold below one
would classify Germany as determinate at the permissive end; with the population itself, determinacy sits at the compensation end only.
In the calibration the density of $m_{\mathrm{crit}}$ steps up as each engagement type enters, so
part (iii) does not apply to it; that the England map has exactly three fixed points,
$\{1,\, 1.77,\, 4\}$, is verified numerically.

\begin{corollary}[Feasibility at the fixed points]
\label{cor:feas}
At the compensation fixed point the survival condition \eqref{eq:survival} fails for every
$\xi \ge 0$ and every $\theta_0 < 1$: with $M = m^{*}_{L} = 1$ it requires
$\xi\kbar \le 1 - 1/\theta_0 < 0$. The compensation fixed point is not a market with weak
signalling but a market with none, at any level of model sharing and any provability short of
certainty.
\end{corollary}

\section{Hysteresis, the watershed, and the doctrine trap}
\label{sec:hyst}

The location parameter $c$ is the mean case-level separation threshold, and by \eqref{eq:mcrit}
it moves with the market primitives: $c$ rises when base provability $\theta_0$ falls or when
model sharing $\xi\kbar$ rises. Improving AI, in the compositional sense of the companion paper,
is therefore an upward drift in $c$, and the comparative statics of $m^{*}(c)$ are the
comparative statics of doctrinal survival under capability growth.

\begin{proposition}[Fold and hysteresis]
\label{prop:fold}
Let the case population shift as a location family, $m_{\mathrm{crit}} + (c - c_0)$, so that $c$
remains its mean, and let $\lbar > \lbar^{\dagger}$ at $c_0$. There exist
$\underline c < c_0 < \bar c$ such that $D$ has at least three fixed points for
$c \in (\underline c, \bar c)$; if the density of $m_{\mathrm{crit}}$ is strictly unimodal, it has
exactly three there, one for $c \notin [\underline c, \bar c]$, and two at each fold, where two of
them merge. On the interval $m^{*}(c)$ is set-valued and the selected equilibrium depends on the
initial condition; at $c = \bar c$ the upper branch merges with the watershed and vanishes, and
$m^{*}$ jumps down to the compensation branch, from which a subsequent reduction of $c$ does not
recover it while $c > \underline c$. In the baseline calibration $\underline c = 1.66$,
$\bar c = 3.05$, and at $c = 2.5$ the two branches differ by $3.00$. The calibrated density steps
up as each engagement type enters and is not unimodal; counted directly, the map has one fixed
point outside $[\underline c, \bar c]$ and exactly three inside, except on a window of width
$1.5 \times 10^{-4}$ at $c = 1.6644$, where the density of the first engagement type falls below
$1/\lbar$ before the second enters and a stable and an unstable fixed point appear within $0.03$
of compensation.
\end{proposition}

The lower fold sits at the edge of what Definition~\ref{def:mcrit} admits. The location family
is a device for comparative statics: it moves the lower end of the population to
$\underline m = 1.4706 + (c - c_0)$, which reaches one at $c = 1.6643$, and below that some cases
would separate at compensation, which requires $\theta_0 \ge 1$. The lower fold,
$\underline c = 1.6635$, lies $0.0009$ beyond that edge, so no admissible population reaches it.
Once at compensation, a fall in the mean threshold that keeps provability below certainty does
not by itself restore the permissive doctrine; Section~\ref{sec:measure} takes up the stochastic
way back.

The reading is that the doctrine, exposed to a slow deterioration in verifiability, does not
degrade smoothly. It holds at the permissive fixed point while $c$ drifts up through the
multiplicity interval, and then fails discontinuously at $\bar c$. Figure~\ref{fig:d1}(b) is the
picture. And because the branches do not coincide, the observable multiple carries no information
about how close the market is to the fold: a jurisdiction at $m^{*} = 4.0$ with $c = 3.0$ and one
with $c = 2.0$ look identical in the contract corpus and are one step and one long march from
collapse respectively.

\begin{corollary}[The doctrine trap]
\label{cor:trap}
Two jurisdictions with identical market primitives $(\theta_0, \xi, \kbar, \ell)$, identical
case populations $F$ and identical uplift $\lbar > \lbar^{\dagger}$ can support permanently
different multiples, and hence permanently different market structures, purely through the
historical initial condition $m_0$. The jurisdiction with $m_0 < m_u$ converges to
$m^{*}_{L} = 1$ and, by Corollary~\ref{cor:feas}, sustains no separating market at any $\xi$ and
any $\theta_0 < 1$.
\end{corollary}

Corollary~\ref{cor:trap} is the paper's answer to a question the companion paper raises and does
not press: why the common-law cells of its regime map are ``already at the boundary''. The
answer here is not that common law is inhospitable to supra-compensatory exposure --- after
\emph{Cavendish} it demonstrably is not --- but that a doctrine which has spent a century
anchored at compensation starts below its own watershed.

\begin{figure}[t]
\centering
\includegraphics[width=\textwidth]{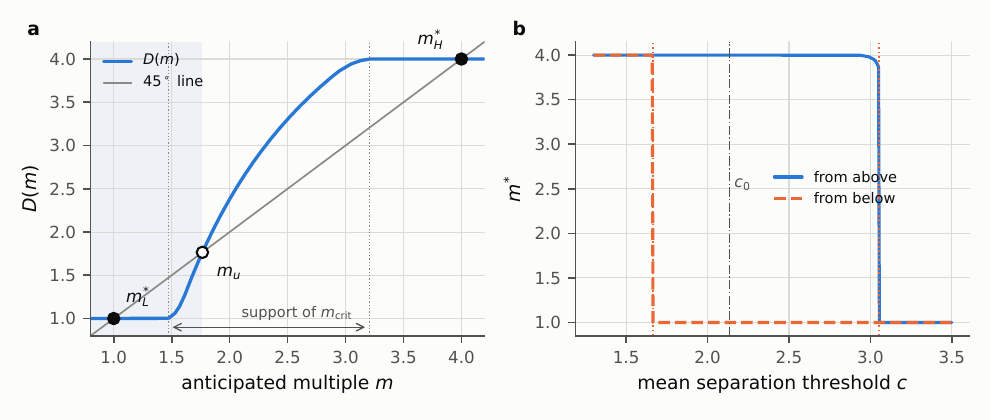}
\caption{(a) The doctrine map $D$ at the baseline calibration $\lbar = 3$, with the case
population of Section~\ref{sec:calib} (mean $c_0 = 2.135$). Filled circles are the stable fixed
points $m^{*}_{L} = 1$ and $m^{*}_{H} = 1 + \lbar = 4$, the open circle is the watershed
$m_u = 1.77$; the shaded region is the basin of compensation. Dotted lines bound the support of
$m_{\mathrm{crit}}$; outside it $D$ is flat. (b) Hysteresis: $m^{*}(c)$ traced from above and
from below as the case population shifts in location, with the folds $\underline c = 1.66$ and
$\bar c = 3.05$ (dotted) and $c_0$ (dash-dot); the lower fold lies at the edge of the admissible
range, where the population's lower end reaches one. Source: \texttt{make\_doctrine.py}.}
\label{fig:d1}
\end{figure}

\section{Ceiling complementarity: when a capital shock becomes a legal one}
\label{sec:complement}

\subsection{The constrained map}

Assumption~\ref{as:recog} says the court recognises the interest when the commitment would
separate. Separation requires $\thf \ge 1/M$ with $M = \min\{\ell, m\}$, not $\thf \ge 1/m$.
Where the solvency ratio binds, no commitment separates however permissive the doctrine, so the
court sees no interest to protect. Hence:

\begin{proposition}[The doctrine inherits the solvency cap]
\label{prop:constrained}
Under Assumptions~\ref{as:interest}--\ref{as:hetero} and $M = \min\{\ell, m\}$, the doctrine map
becomes
\begin{equation}
D(m;\ell) \;=\; 1 \;+\; \lbar\, F\big(\min\{\ell,\, m\}\big),
\label{eq:Dl}
\end{equation}
which is nondecreasing in $m$ and constant on $[\ell, 1+\lbar]$.
\end{proposition}

\subsection{The collapse threshold}

\begin{proposition}[Collapse threshold and discontinuity]
\label{prop:collapse}
Let $m^{*}_{L} < m_u < m^{*}_{H}$ be the fixed points of the unconstrained map $D$. Then the
equilibrium ceiling under \eqref{eq:Dl}, selected from the upper basin, is
\begin{equation}
M(\ell) \;=\;
\begin{cases}
m^{*}_{H}, & \ell \ge m^{*}_{H},\\[2pt]
\ell, & m_u \le \ell < m^{*}_{H},\\[2pt]
m^{*}_{L}, & m^{*}_{L} < \ell < m_u,\\[2pt]
\ell, & \ell \le m^{*}_{L}.
\end{cases}
\label{eq:Mofl}
\end{equation}
$M$ is discontinuous at $\ell = m_u$, with a downward jump of $m_u - m^{*}_{L}$. On the third
branch the binding ceiling is the doctrine at a level strictly below the capital constraint that
produced it.
\end{proposition}

The threshold is not a new parameter. It is the watershed of the unconstrained doctrine map, and
it is therefore determined by the doctrine's own structure rather than by anything about
insurance. In the baseline calibration $m_u = 1.77$ and the jump is $0.77$; the critical sharing
$\xi^{*}$ at $\theta_0 = 0.85$ falls from $0.333$ at $\ell = m_u = 1.77$ to zero immediately below.
Table~\ref{tab:complement} traces it. The numerical iteration of \eqref{eq:Dl} reproduces
\eqref{eq:Mofl} to machine precision, which is the verification we report rather than a
robustness check.

\begin{table}[t]
\centering\small
\caption{The equilibrium ceiling under exogenous and endogenous doctrine, England calibration
($\lbar = 3$, case population of Section~\ref{sec:calib}, $m_u = 1.77$), $\theta_0 = 0.85$.
$\xi^{*}$ is evaluated at $\kbar = 1$; zero denotes infeasibility at any $\xi$.}
\label{tab:complement}
\begin{tabular}{lrrrrr}
\toprule
& \multicolumn{2}{c}{exogenous $m = m^{*}_{H}$} & \multicolumn{3}{c}{endogenous $m^{*}(\ell)$}\\
\cmidrule(lr){2-3}\cmidrule(lr){4-6}
$\ell = \Lbar/v$ & $M$ & $\xi^{*}$ & $m^{*}(\ell)$ & $M$ & $\xi^{*}$\\
\midrule
5.00 & 4.00 & 0.706 & 4.00 & 4.00 & 0.706\\
4.00 & 4.00 & 0.706 & 4.00 & 4.00 & 0.706\\
3.00 & 3.00 & 0.608 & 3.91 & 3.00 & 0.608\\
2.50 & 2.50 & 0.529 & 3.30 & 2.50 & 0.529\\
2.10 & 2.10 & 0.440 & 2.60 & 2.10 & 0.440\\
2.00 & 2.00 & 0.412 & 2.38 & 2.00 & 0.412\\
1.90 & 1.90 & 0.381 & 2.14 & 1.90 & 0.381\\
1.80 & 1.80 & 0.346 & 1.87 & 1.80 & 0.346\\
\addlinespace
1.50 & 1.50 & 0.216 & 1.00 & 1.00 & 0.000\\
1.00 & 1.00 & 0.000 & 1.00 & 1.00 & 0.000\\
\bottomrule
\end{tabular}
\end{table}

\begin{figure}[t]
\centering
\includegraphics[width=\textwidth]{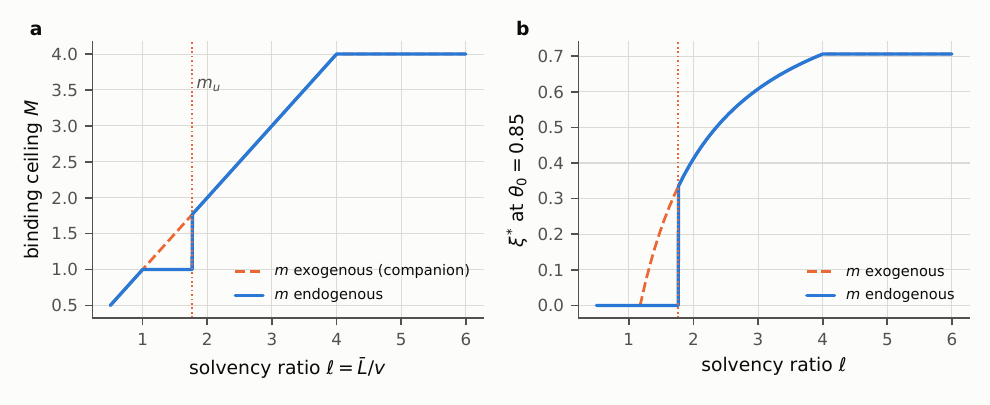}
\caption{(a) The binding ceiling $M$ as a function of the solvency ratio $\ell$, with $m$
exogenous at $m^{*}_{H}$ (dashed, the companion paper's benchmark) and endogenous (solid). The
discontinuity is at $\ell = m_u = 1.77$. (b) The critical sharing $\xi^{*}$ implied by each, at
$\theta_0 = 0.85$, $\kbar = 1$. Source: \texttt{make\_doctrine.py}.}
\label{fig:d2}
\end{figure}

\subsection{Irreversibility}

\begin{corollary}[The January 2026 shock does not reverse]
\label{cor:irrev}
Let $\ell$ fall from $\ell_0 > m^{*}_{H}$ to $\ell_1 < m_u$ for $s$ periods and then return to
$\ell_0$, under the adjustment of Proposition~\ref{prop:stab} from $m^{*}_{H}$. While
$m_t \ge \ell_1$ the constrained map is flat at $D(\ell_1) < m_u$, so the multiple is below the
watershed from period
\begin{equation}
t^{*}(\ell_1) \;=\; \left\lfloor \frac{\ln\!\big[(m_u - D(\ell_1))/(m^{*}_{H} - D(\ell_1))\big]}
{\ln(1-\eta)} \right\rfloor + 1
\label{eq:tstar}
\end{equation}
on. If $s \ge t^{*}(\ell_1)$, the multiple converges to $m^{*}_{L}$ after $\ell$ is restored,
because it lies in the lower basin of the unconstrained map; if $s < t^{*}(\ell_1)$, it returns
to $m^{*}_{H}$. A shock with $\ell_1 \ge m_u$ is reversed at any duration. In the calibration,
with $\ell_0 = 5$, $\ell_1 = 1.5$ and $\eta = 0.5$, $t^{*} = 3$: the multiple falls from $4.00$
below the watershed in the third period and stands at $1.00$ thirty periods after cover is
restored, while a one-period shock leaves it at $2.51$ at restoration, from where it returns to
$4.00$.
\end{corollary}

Corollary~\ref{cor:irrev} bears on a
falsification condition of the companion paper. That paper names, as the fourth condition that
would overturn its regime map, the emergence of a deep affirmative market in AI liability cover
with capacity well above engagement values, which would make $\ell$ slack and return the doctrine
channel to dominance. Under an endogenous doctrine that condition is necessary and not
sufficient: capacity restored after the doctrine has crossed its watershed returns the market to
$M = m^{*}_{L}$, not to $M = \min\{\ell, m^{*}_{H}\}$. The depth of the shock decides whether
it can do lasting damage, and its duration whether it does. The fuse is two periods for a deep
shock and lengthens without bound as $\ell_1$ approaches the watershed from below, because
$D(\ell_1) \to m_u$: eight periods at $\ell_1 = 1.76$, fourteen at $1.765$
(Figure~\ref{fig:d3}(b)).

\begin{figure}[t]
\centering
\includegraphics[width=0.85\textwidth]{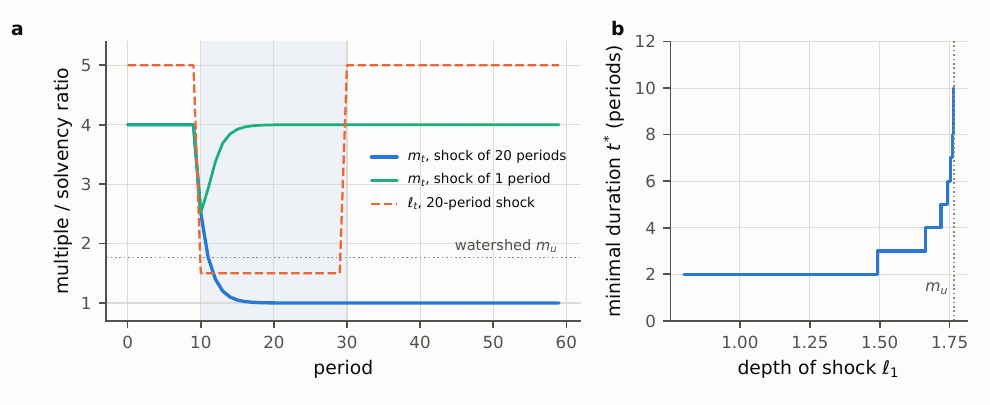}
\caption{Irreversibility and its fuse. (a) The generative-AI exclusion attaches in period 10 and
is withdrawn in period 30 (shaded, $\ell_1 = 1.5$); the solvency ratio returns to its pre-shock
level, the doctrine does not. The same exclusion withdrawn after one period is fully reversed.
(b) The minimal duration $t^{*}(\ell_1)$ of \eqref{eq:tstar} after which a shock of depth
$\ell_1$ is irreversible; shocks with $\ell_1 \ge m_u = 1.77$ are reversed at any duration.
Source: \texttt{make\_doctrine.py}.}
\label{fig:d3}
\end{figure}

\section{Stochastic selection and metastable irreversibility}
\label{sec:stoch}

The deterministic model leaves two questions open. Selection between $m^{*}_{L}$ and $m^{*}_{H}$ is
by initial condition, and Corollary~\ref{cor:irrev} states an absolute irreversibility, which is
the deterministic limit of a stochastic statement. This section answers both by perturbing the adjustment
and applying the one-dimensional Freidlin--Wentzell apparatus, where everything is exact.

\begin{assumptionS}\label{as:noise}
Anticipated multiples follow the perturbed adjustment
\begin{equation}\label{eq:sde}
dm_t \;=\; \eta\,\big(D(m_t) - m_t\big)\,dt \;+\; \sigma\, dW_t,
\qquad m_t \in [1,\, 1+\lbar],
\end{equation}
with reflection at the boundaries, $\eta \in (0,1]$ the adjustment speed of
Proposition~\ref{prop:stab} and $\sigma > 0$ a noise intensity.
\end{assumptionS}

The perturbation is a reduced form, and we say what it is a reduced form of. It is not the limit
of idiosyncratic revisions. If $n$ drafting parties each revise toward the doctrine's response to
the population mean with independent logit noise, as in the protocol of \citet{blume1993}, the
population mean follows the deterministic adjustment of Proposition~\ref{prop:stab} up to noise
of order $\sigma/\sqrt{n}$, which vanishes in the mean-field limit; with $n = 400$ drafting
parties it is a twentieth of the noise in \eqref{eq:sde}. The noise in
\eqref{eq:sde} is common noise: coordinated revisions of standard forms, professional-body model
clauses, an appellate remark that moves many drafters at once. That is why $\sigma$ is to be
read off the year-on-year dispersion of agreed multiples within a clause family
(Section~\ref{sec:falsify}) and not derived from the number of drafting parties. A population of
drafting parties with logit noise can look like a microfoundation of \eqref{eq:sde}; it is not one.

\subsection{The doctrinal potential}

\begin{proposition}[Potential and invariant density]\label{prop:potential}
Define the doctrinal potential
\begin{equation}\label{eq:potential}
U(m) \;=\; -\,\eta \int_{1}^{m} \big(D(u) - u\big)\, du .
\end{equation}
The drift of \eqref{eq:sde} is $-U'$; the process is a gradient diffusion; its unique invariant
density is $\pi_{\sigma}(m) \propto e^{-2U(m)/\sigma^{2}}$; and the local minima of $U$ are
exactly the stable fixed points of $D$, the local maxima its unstable ones. In the baseline both
wells lie at the ends of the doctrinal range, which are the reflecting boundaries of
\eqref{eq:sde}: the compensation well at $U(m^{*}_{L}) = U(1) = 0$, the watershed at
$U(m_u) = 0.0990$, and the permissive well at $U(m^{*}_{H}) = U(4) = -0.5476$. The barrier out of
the compensation basin is $\Delta U_{L} = 0.0990$ and out of the permissive basin
$\Delta U_{H} = 0.6467$, an asymmetry of $6.5$ to one in favour of the permissive doctrine.
\end{proposition}

\subsection{The Maxwell rule and a closed-form tipping point}

\begin{theorem}[Maxwell selection]\label{thm:maxwell}
Let $D$ have at least three fixed points, let $m^{*}_{L} < m^{*}_{H}$ be the least and the
greatest, both stable, and $m_u$ the watershed between them, and define the signed area
\begin{equation}\label{eq:area}
A(c) \;=\; \int_{m^{*}_{L}(c)}^{m^{*}_{H}(c)} \big(D(u;c) - u\big)\, du
\;=\; -\,\frac{U(m^{*}_{H}) - U(m^{*}_{L})}{\eta}.
\end{equation}
(i) As $\sigma \to 0$, $\pi_{\sigma}$ concentrates on the global minimiser of $U$ among
$\{m^{*}_{L}, m^{*}_{H}\}$: the permissive doctrine is stochastically selected if and only if
$A(c) > 0$. (ii) $A > 0$ at the lower fold and $A < 0$ at the upper fold, so a Maxwell point
$c^{\ast}$ with $A(c^{\ast}) = 0$ lies strictly between the folds; on every interval of $c$ on
which the outer fixed points move continuously, $A'(c) = -(m^{*}_{H} - m^{*}_{L}) < 0$, and in
the calibration $c^{\ast}$ is unique. (iii) On every interval
of $c$ on which the support of $m_{\mathrm{crit}}$ lies inside $(1, 1+\lbar)$, the outer fixed
points are $1$ and $1 + \lbar$ and $A(c) = \lbar\,(1 + \lbar/2 - c)$. Hence, if the support lies
inside the doctrinal range at $c = 1 + \lbar/2$,
\begin{equation}\label{eq:cstar}
\boxed{\; c^{\ast} \;=\; 1 + \frac{\lbar}{2} \;}
\end{equation}
exactly: the long-run tipping point is reached when the mean case-level threshold reaches the
midpoint of the doctrinal range $[1,\, 1+\lbar]$, whatever the shape, dispersion or skewness of
the case distribution. (iv) The same closed form holds for a case distribution with unbounded
support that is symmetric about its mean --- logistic, normal, or any $S$ with
$S(z) + S(-z) = 1$ --- the logistic benchmark of Table~\ref{tab:rob}.
\end{theorem}

The proof of (iii) is two lines. With the support inside the doctrinal range, $D \equiv 1$
below it and $D \equiv 1 + \lbar$ above it, so the outer fixed points are the ends of the range,
and $A(c) = \int_1^{1+\lbar} (1 + \lbar F(u) - u)\,du = \lbar(1 + \lbar - c) - \lbar^{2}/2$,
because $\int_1^{1+\lbar} F = 1 + \lbar - c$ for a distribution with mean $c$ on that range.
No symmetry is used. Numerically the bisection root agrees with \eqref{eq:cstar} to machine
precision (recorded deviation $0.0$), $A(c)$ agrees with the linear form on the whole interval
$c \in (1.664, 2.930)$ on which the calibrated population lies inside $(1, 4)$ (maximal deviation
$0.0$), and the symmetric benchmarks of (iv) give $A(c^{\ast}) = 0.0$ for the logistic and the
sd-matched normal.

In the baseline, $A(2.13) = 1.10 = 3 \times 0.37 > 0$: England's permissive doctrine is not
merely historically selected but selected in the long run, sitting $0.37$ below the Maxwell
point. The margin, however, defines a new object.

\begin{definition}[Borrowed time]\label{def:borrowed}
For $c \in (c^{\ast}, \bar c)$ --- between the Maxwell point $2.50$ and the upper fold $3.05$, a
band of width $0.55$ --- the permissive fixed point exists, is locally stable, and is not selected
in the long run: $U(m^{*}_{H}) > U(m^{*}_{L})$, so the invariant distribution concentrates on the
compensation doctrine as $\sigma \to 0$.
\end{definition}

Whether a doctrine in this band is abandoned on any horizon that matters is a separate question,
and the answer is mostly no.

\begin{proposition}[Policy-horizon threshold]\label{prop:cdagger}
For noise $\sigma$ and horizon $T$, let $c^{\dagger}(\sigma, T)$ be the location at which the
expected exit time of the permissive doctrine to the watershed equals $T$. The barrier
$\Delta U_H$ falls strictly in $c$ on $(c^{\ast}, \bar c)$ and vanishes at $\bar c$, and so does
the expected exit time; hence, for every $\sigma$ at which the expected exit time of the
permissive doctrine at $c^{\ast}$ exceeds $T$, $c^{\ast} < c^{\dagger}(\sigma, T) < \bar c$, the doctrine is expected to be
abandoned within $T$ only for $c > c^{\dagger}(\sigma, T)$, and $c^{\dagger}(\sigma, T) \to \bar c$
as $\sigma \to 0$. At the Maxwell point the two barriers are equal,
$\Delta U_H = \Delta U_L = 0.291$ in the calibration --- $2.9$ times the barrier $0.099$ that makes
the collapse metastably irreversible at the baseline --- and the expected exit time of the
permissive doctrine is $10^{25.7}$, $10^{6.8}$ and $10^{3.8}$ periods at $\sigma = 0.10$, $0.20$
and $\sigma^{\dagger}(30) = 0.28$. The thresholds are $c^{\dagger}(\sigma, 30) = 3.01$, $2.92$
and $2.83$ at those noise levels, and $c^{\dagger}(\sigma, 10) = 3.03$, $2.98$ and $2.92$.
\end{proposition}

By the standard the paper applies to the collapsed doctrine, a permissive doctrine at the Maxwell
point is irreversibly permissive. The Maxwell point marks where the very-long-run selection
flips; the threshold at which a permissive doctrine is actually lost on a policy horizon lies
close to the fold.

\begin{theorem}[Form-free selection]\label{thm:formfree}
For any continuous drift $\eta(D - \cdot)$ with finitely many hyperbolic fixed points --- no
functional-form assumption, any number of fixed points --- the stochastically selected state is
the global minimiser of the potential \eqref{eq:potential}, provided the minimiser is unique; if
several wells tie, the invariant distribution splits between them. Under the bimodal heterogeneity mixture
of Section~\ref{sec:falsify} (weights $\tfrac12$, centres $1.55$ and $2.85$, scale
$0.085$) the map has five fixed points, $\{1.00, 1.49, 2.54, 2.68, 4.00\}$, of which three are stable; the
potential has three wells, $U = (-0.00, -0.15, -0.45)$, and selection is by the global minimum, here the
permissive doctrine at $4.0$.
\end{theorem}

Theorem~\ref{thm:formfree} answers the selection question independently of the functional form: the
\emph{location} of thresholds is form-dependent, but the \emph{selection rule} --- minimise the
doctrinal potential --- is not.

\subsection{Metastable irreversibility}

\begin{proposition}[Exact recovery time and Kramers asymptotics]\label{prop:mfpt}
Let the state start at $m^{*}_{L}$ after the shock, with cover restored. The expected first
passage to the watershed $m_u$ (entry to the permissive basin's boundary; the well-to-well time
is asymptotically twice this) is exactly
\begin{equation}\label{eq:mfpt}
\E[\tau] \;=\; \frac{2}{\sigma^{2}} \int_{m^{*}_{L}}^{m_u}
e^{2U(y)/\sigma^{2}} \int_{1}^{y} e^{-2U(z)/\sigma^{2}}\, dz\, dy
\;\;\sim\;\;
\frac{\pi}{2\sqrt{U''(m^{*}_{L})\,\lvert U''(m_u)\rvert}}\;
e^{2\Delta U_{L}/\sigma^{2}},
\end{equation}
with $U''(m^{*}_{L}) = \eta = 0.5$ exactly, because $D$ is flat at compensation,
$\lvert U''(m_u)\rvert = 0.983$ and Kramers prefactor $2.24$. The factor $1/2$ relative to the
interior-well formula is the half-Gaussian of a well that sits on the reflecting boundary.
\end{proposition}

\begin{corollary}[Metastable irreversibility]\label{cor:meta}
Corollary~\ref{cor:irrev} is the $\sigma \to 0$ limit of the following statement: after the
shock, with cover restored, the doctrine does return to the permissive basin --- the expected
first passage to the watershed is $10^{34.7}$ periods at $\sigma = 0.05$, $10^{8.9}$ at
$\sigma = 0.10$, $10^{4.2}$ at $\sigma = 0.15$ and $10^{2.5}$ at $\sigma = 0.20$. The passage
time is approximately exponential, so the probability of a return within a horizon $T$ is
$1 - e^{-T/\E[\tau]}$. Let $\sigma^{\dagger}(T)$ be the largest $\sigma$ with
$\E[\tau] \ge T$, and $\sigma^{\ddagger}(T)$ the largest with a return probability within $T$ of
at most ten per cent, $\E[\tau] \ge 9.49\,T$. Then $\sigma^{\dagger}(10) = 0.35$,
$\sigma^{\dagger}(30) = 0.28$, $\sigma^{\dagger}(100) = 0.23$, and $\sigma^{\ddagger}(10) = 0.23$,
$\sigma^{\ddagger}(30) = 0.20$, $\sigma^{\ddagger}(100) = 0.18$. At $\sigma = \sigma^{\dagger}(T)$
the doctrine returns within $T$ with probability $1 - e^{-1} = 0.63$; the deterministic reading
is reliable on horizon $T$ only below $\sigma^{\ddagger}(T)$.
\end{corollary}

The exact quadrature in \eqref{eq:mfpt} is overflow-safe at any $\sigma$. Because the
compensation well now sits on the reflecting boundary rather than just above it, the half-Gaussian
Kramers form is accurate to leading order: it differs from the quadrature by $0.01$--$0.03$ in
natural logs over $\sigma \in [0.05, 0.30]$. Direct Euler--Maruyama simulation at
$\sigma \in \{0.20,\allowbreak 0.25,\allowbreak 0.30\}$, with every path absorbed, reproduces
the quadrature within $0.06$--$0.07$ in natural logs, the $O(dt)$ discretisation bias at
$dt = 0.005$, and confirms the exponential law: $63$--$64$ per cent of passages occur before the
mean and the coefficient of variation is $0.94$--$1.04$. Figure~\ref{fig:d5} shows all three on
the Arrhenius plot.

\subsection{The monoculture bridge}

The location $c$ is not free: it is the mean of the case-level separation threshold
$m_{\mathrm{crit}} = 1/(\theta_0(1 - \xi\kbar))$, so rising model monoculture moves the doctrine
toward its thresholds. Shifting the sharing range $[0.2, 0.6]$ of the companion paper's grid
upward moves the whole case population, and because $1/(1 - \xi\kbar)$ is convex it raises the
dispersion of $m_{\mathrm{crit}}$ with its mean, from a standard deviation of $0.44$ to $0.72$ at
the fold. Along that path the population crosses the Maxwell point when mean sharing reaches
$\xi\kbar = 0.48$ and the fold when it reaches $0.52$. On a policy horizon the
relevant crossing is $c^{\dagger}$: along the same path the permissive doctrine is expected to be
abandoned within thirty periods from mean sharing $0.51$--$0.52$, depending on $\sigma$ ---
essentially at the fold. The long-run selection flips four points of mean sharing before the
fold; the observable collapse does not come before it. Panel~(a) of Figure~\ref{fig:d6} maps the
recovery time of the compensation doctrine over $(c, \sigma)$, panel~(b) the survival time of
the permissive doctrine, whose ten- and thirty-period contours are $c^{\dagger}(\sigma, T)$.

Homogenisation moves two moments of the case population, and they have to be kept apart. A
rise in the level of model sharing raises both the mean and the dispersion of
$m_{\mathrm{crit}}$, as above. A compression of its spread at a given mean lowers the mean,
because $m_{\mathrm{crit}}$ is convex in sharing: compressing the sharing range around $0.40$
from $[0.20, 0.60]$ to $[0.25, 0.55]$, $[0.30, 0.50]$ and $[0.35, 0.45]$ lowers $c$ from $2.135$
to $2.098$, $2.073$ and $2.058$, and moves the mean sharing at which the Maxwell point is reached
from $0.480$ to $0.492$, $0.500$ and $0.506$. The Maxwell point in $c$ is untouched ---
Theorem~\ref{thm:maxwell}(iii) involves the mean only --- but the warning level in mean sharing is
not.

\paragraph{Correlated revision and jump risk.}
Assumption~S treats common revision shocks as continuous and Gaussian, and real drafting is
lumpier: a small number of large firms and professional-body model clauses coordinate changes, so
the perturbation has a jump component --- a single authoritative revision of a standard form moves
much of the field at once. Two consequences follow. First, a jump of size
$m_u - m^{*}_{L} = 0.77$ or more
clears the barrier in one step, replacing the Kramers time by the arrival rate of such
revisions; the falsification threshold $\sigma^{\dagger}(T)$ of Corollary~\ref{cor:meta}
bounds the diffusion component only, and a jump-diffusion estimate of recovery is bounded above
by the first arrival of a sufficiently large coordinated revision. Second, this is not a defect
of the policy analysis but its mechanism: the instrument we call doctrinal clarification ---
appellate guidance, a restatement, a model-clause endorsement --- \emph{is} such a jump, and it
works precisely because it is not diffusion. What the jump channel does qualify is the
selection theorem: for L\'evy perturbations the selected state need not minimise $U$, and we
flag the large-deviations analysis under jumps as an open refinement rather than assert it.

\begin{figure}[t]
\centering
\includegraphics[width=\textwidth]{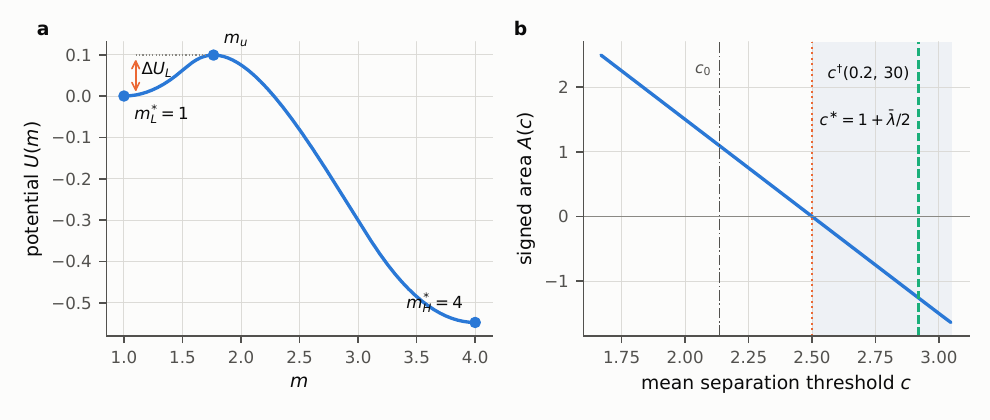}
\caption{The doctrinal potential and the Maxwell point. Panel (a): $U(m)$ in the baseline, with
the two wells at the ends of the doctrinal range, the watershed, and the escape barrier
$\Delta U_{L}$. Panel (b): the signed area $A(c)$, linear while the case population lies inside
the doctrinal range, with its zero at the closed-form Maxwell point $c^{\ast} = 1 + \lbar/2 =
2.50$, the baseline $c_0 = 2.13$, the policy-horizon threshold $c^{\dagger}(0.2, 30) = 2.92$, and
the band (shaded) between $c^{\ast}$ and the fold in which the permissive doctrine is locally
stable but not selected in the long run. Source: \texttt{make\_stochastic.py}.}
\label{fig:d4}
\end{figure}

\begin{figure}[t]
\centering
\includegraphics[width=\textwidth]{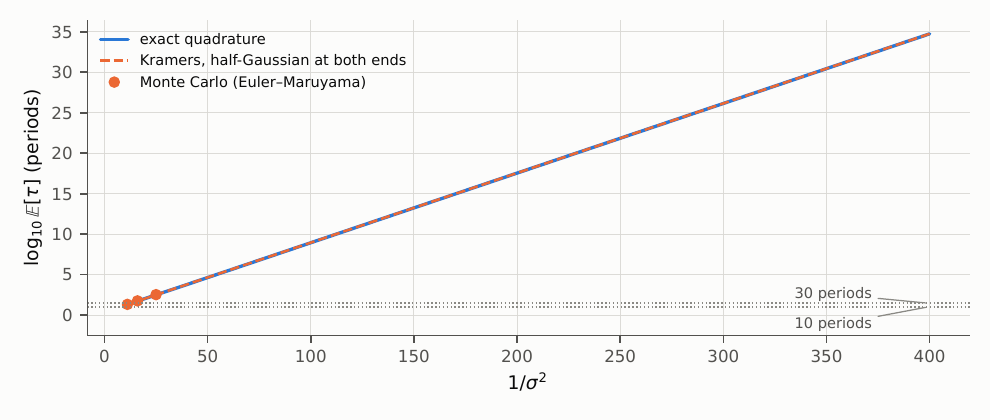}
\caption{Metastable irreversibility on the Arrhenius plot: $\log_{10}$ expected recovery time
against $1/\sigma^{2}$. Exact quadrature (solid), Kramers asymptotics with half-Gaussians at the
well and the saddle (dashed; indistinguishable at this scale), Euler--Maruyama Monte Carlo
(points; every path absorbed). Horizontal dotted lines mark ten- and thirty-period policy
horizons; their intersections with the exact curve define $\sigma^{\dagger}(10) = 0.35$ and
$\sigma^{\dagger}(30) = 0.28$. Source: \texttt{make\_stochastic.py}.}
\label{fig:d5}
\end{figure}

\begin{figure}[t]
\centering
\includegraphics[width=\textwidth]{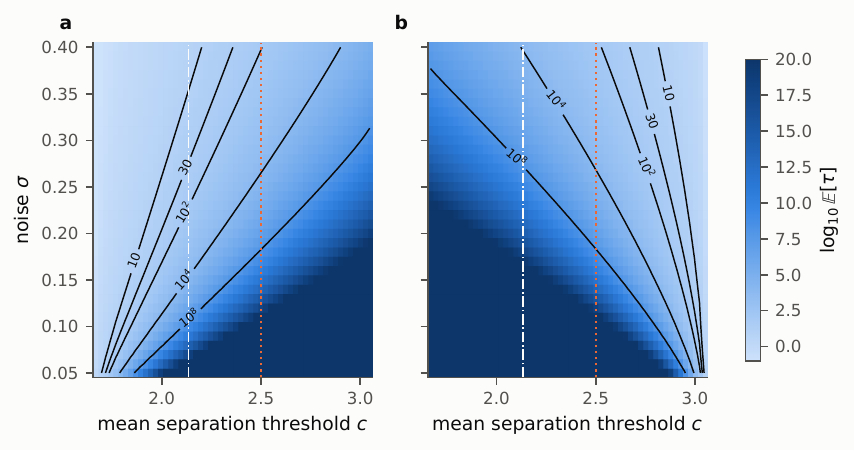}
\caption{Two clocks over the $(c, \sigma)$ plane, both on a $\log_{10}$ scale with contours at
$10$, $30$, $10^{2}$, $10^{4}$ and $10^{8}$ periods, across the fold interval. Panel (a): expected
recovery time out of the compensation doctrine. Panel (b): expected exit time of the permissive
doctrine; its ten- and thirty-period contours are the policy-horizon thresholds
$c^{\dagger}(\sigma, T)$. The Maxwell point is dotted, the baseline $c_0$ dash-dotted. Source:
\texttt{make\_stochastic.py}.}
\label{fig:d6}
\end{figure}

\section{Calibration}
\label{sec:calib}

\subsection{The case population}

By Definition~\ref{def:mcrit} the case population is not free: it is the distribution of
$m_{\mathrm{crit}} = 1/(\theta_0(1-\xi\kbar))$ over the cases a court sees. We take it from the
companion paper's own calibration grid.
Engagement-type verifiabilities are $\theta_0 \in \{0.80, 0.82, 0.85, 0.78\}$ --- a scoped legal
matter, a statutory audit, a fixed-price software project, a teleradiology read --- drawn with
equal probability, and $\xi\kbar$ is drawn uniformly on $[0.20, 0.60]$, which spans the
mid-range $\xi \in [0.2,0.4]$ on which that paper organises its calibration together with the
upper tail of the sharing range at the conservative benchmark $\kbar = 1$. The distribution
function is available in closed form,
\[
F(m) \;=\; \frac14 \sum_{\theta_0} \min\Big\{1,\, \max\Big\{0,\,
\frac{1 - 1/(\theta_0 m) - 0.2}{0.4}\Big\}\Big\},
\]
so every number below is computed exactly rather than by simulation:
\[
c \;=\; \E[m_{\mathrm{crit}}] \;=\; 2.135,
\qquad
\mathrm{sd}(m_{\mathrm{crit}}) \;=\; 0.436,
\qquad
\mathrm{supp}\, m_{\mathrm{crit}} \;=\; [1.471,\; 3.205],
\]
with median $2.053$ and skewness $0.50$. The population is right-skewed and bounded away from one.
A logistic law with the same mean and standard deviation
($c = 2.135$, scale $0.240$) places $5.9$ per cent of cases below the smallest threshold the
grid can produce and $0.9$ per cent below one. That tail creates a low fixed point at $1.03$ where
the population has one at exactly $1$, a watershed at $1.95$ where the population has one at
$1.77$, and a determinate Germany where the population has a watershed at $1.53$. The logistic is
kept in Table~\ref{tab:rob} as a benchmark and does not enter any other number.

\subsection{The doctrinal parameter \texorpdfstring{$\lbar$}{lambda-bar} and a jurisdictional taxonomy}

The uplift $\lbar$ is read off primary sources, one jurisdiction at a time.

\emph{United States.} Restatement (Second) of Contracts \S 356 and UCC \S 2-718 test agreed
damages against anticipated or actual loss. There is no interest limb, so $\lbar = 0$ and
$D \equiv 1$.

\emph{Austria.} \S 1336(2) ABGB mandates judicial moderation, on the prevailing view including
between undertakings, and the companion paper reports a range $1.2$--$2$. We set $\lbar = 1$:
\S 1336(2) is a moderation right exercised \emph{ex post}, not an interest-referenced uplift ---
the reported band measures multiples that survived moderation, not a doctrinal premium granted
for a protected interest, and conflating the two would double-count. The sensitivity to this
choice is reported in Section~\ref{sec:regime} and is the taxonomy's weakest point.

\emph{England and Wales.} \emph{Cavendish} at [32] states the interest-referenced test, and the
\emph{Houssein} chain supplies the only verified post-2015 magnitude: on remission a default rate
at four times the standard rate was held non-penal ([2025] EWHC 2749 (Ch)), affirmed on second
appeal ([2026] EWCA Civ 830). We set $1 + \lbar = 4$, hence $\lbar = 3$. This is an
\emph{observed non-penal multiple}, not an estimated ceiling; Section~\ref{sec:falsify} treats the
resulting selection problem.

\emph{India.} \emph{BPL Ltd v Morgan Securities and Credits Pvt Ltd}, 2025 INSC 1380 (4 December
2025), adopts the \emph{Cavendish} legitimate-interest test under \S 74 of the Contract Act,
displacing the genuine-pre-estimate orthodoxy of \emph{Kailash Nath}, and upholds a contractual
rate of 36 per cent with monthly rests as protecting a legitimate commercial interest between
sophisticated parties. We assign India the English parameters, which makes it the natural
experiment discussed below.

\emph{Germany, negotiated B2B.} \S 348 HGB withholds the \S 343 BGB reduction from merchants
acting in the course of business, leaving only \S 242 BGB; the companion paper reports a range
$3$--$25$ for the negotiated channel. We set $\lbar = 24$. The standard-terms channel is a
different object --- BGH VII ZR 42/22 caps a construction penalty near five per cent of the
contract sum --- and belongs with the American parameters.

\begin{table}[t]
\centering\small
\caption{Jurisdictional taxonomy with the case population of Section~\ref{sec:calib},
$\lbar^{\dagger} = 1.95$. ``Determinate'' means compensation is the unique fixed point;
``multiple'' means three. $\xi^{*}$ at the greatest fixed point is evaluated at
$\theta_0 = 0.85$, $\kbar = 1$; zero denotes infeasibility at any $\xi$.}
\label{tab:jur}
\begin{tabular}{@{}llrrlrr@{}}
\toprule
Jurisdiction & Source of $\lbar$ & $\lbar$ & $\lbar/\lbar^{\dagger}$ & regime &
$m_u$ & $\xi^{*}(m^{*}_{H})$\\
\midrule
United States & Restatement \S 356 & 0 & 0.00 & determinate & --- & 0.000\\
Austria & \S 1336(2) ABGB & 1 & 0.51 & determinate & --- & 0.000\\
England \& Wales & \emph{Houssein} $4\times$ & 3 & 1.54 & multiple & 1.77 & 0.706\\
India (post-\emph{BPL}) & 2025 INSC 1380 & 3 & 1.54 & multiple & 1.77 & 0.706\\
Germany (neg.\ B2B) & \S 348 HGB & 24 & 12.30 & multiple & 1.53 & 0.953\\
\bottomrule
\end{tabular}
\end{table}

Determinacy sits at the compensation end only. Where the doctrine admits no
supra-compensatory uplift, or one too small to lift any market above its own threshold --- the
United States and Austria --- compensation is the unique fixed point. Above $\lbar^{\dagger}$
every jurisdiction is indeterminate, and what the size of the uplift controls is the size of the
compensation basin: the watershed is $1.77$ in England and India and $1.53$ in negotiated German
contracting. German doctrine survives every shock that keeps the solvency ratio above $1.53$,
English doctrine only those that keep it above $1.77$. That England, India and Germany carry a
watershed while the United States and Austria do not is a prediction about which jurisdictions
should display clustered contract-corpus multiples rather than a single conventional one.

Austria is the case to watch for robustness. At the calibrated $\lbar = 1$ the ratio
$\lbar/\lbar^{\dagger}$ is $0.51$ and compensation is the unique fixed point. Sweeping $\lbar$
across the reported $1.2$--$2$ band, condition \eqref{eq:threefp} keeps it so up to
$\lbar^{\dagger} = 1.95$ --- at $\lbar \in \{1.2, 1.5, 1.9\}$ the unique fixed point is
compensation --- and only the top of the band, $\lbar \in (1.95, 2]$, crosses into multiplicity,
with fixed points $\{1.00, 2.20, 2.80\}$ at $\lbar = 2$. The Austrian row of
Table~\ref{tab:jur} is therefore conditional on all but the top $0.05$ of its own source band;
an asymmetric application
of the moderation right that let enforced multiples drift toward the top of the band is exactly
the mechanism that would move Austria into the multiple regime, and we flag the row as the
taxonomy's least robust. The Austrian moderation right, at the calibrated $\lbar = 1$, is
sufficient on its own to pin the doctrine at compensation.

\subsection{The regime map re-run}
\label{sec:regime}

We now re-run the companion paper's regime map --- six engagement types across four
jurisdictions, before and after the January 2026 generative-AI insurance exclusions, with own
capital at 15 per cent of the insured limit --- twice: once with $m$ exogenous at that paper's
values ($m_{US} = 1$, $m_{AT} = 1.5$, $m_{UK} = 3$, $m_{DE} = 10$), and once with $m$ at the
equilibrium of \eqref{eq:Dl}. Survival counts are evaluated at $\kbar = \kappa_0 = 0.60$, as
there.

\begin{table}[t]
\centering\small
\caption{The regime map with exogenous and endogenous doctrine, 24 engagement--jurisdiction
cells. The exogenous rows reproduce the companion paper's corollary ``The exclusions move the regime boundary'' exactly. A cell is
solvency-bound if $\ell < m$; four endogenous post-exclusion cells have $\ell = m = 1$ exactly,
are counted as doctrine-bound, and have $\xi^{*} = 0$ either way.}
\label{tab:regime}
\begin{tabular}{lrrrr}
\toprule
& solvency-bound & survive $\xi = 0.2$ & survive $\xi = 0.4$ & infeasible at any $\xi$\\
\midrule
\multicolumn{5}{l}{\emph{exogenous $m$ (companion paper)}}\\
\quad PI cover responds & 3 & 18 & 12 & 6\\
\quad after AI exclusion & 11 & 12 & 4 & 12\\
\addlinespace
\multicolumn{5}{l}{\emph{endogenous $m$ (this paper)}}\\
\quad PI cover responds & 5 & 12 & 12 & 12\\
\quad after AI exclusion & 6 & 4 & 4 & 20\\
\bottomrule
\end{tabular}
\end{table}

Three features of Table~\ref{tab:regime} matter.

\emph{Replication.} With $m$ exogenous our implementation returns three
solvency-bound cells with cover and eleven after exclusion, and survival counts of eighteen and
twelve at $\xi = 0.2$ falling to twelve and four at $\xi = 0.4$ --- the published numbers, to the
cell. Everything that differs below is attributable to endogenising $m$ and to nothing else.

\emph{Polarisation.} Under exogenous $m$ the map contains cells
with intermediate critical sharing --- Austrian cells at $\xi^{*} \approx 0.15$--$0.19$, English
mid-size cells at $\xi^{*} = 0.17$--$0.19$ after exclusion. Under an endogenous doctrine these
vanish: a cell either sits at a fixed point that supports separation robustly
($\xi^{*} = 0.68$--$0.93$ after exclusion) or at one that supports none ($\xi^{*} = 0$). This
is why the endogenous counts at $\xi = 0.2$ and $\xi = 0.4$ coincide when cover responds, at
twelve, and again after exclusion, at four, while the exogenous counts differ. The intermediate cells were artefacts of
holding $m$ at a conventional value that is not an equilibrium of the doctrine that produces it.

\emph{Direction of the change.} It is not uniformly adverse to the companion paper's
conclusion, and the companion paper says as much when it warns that endogenising $m$ ``cuts
against our own conclusion''. In the English column the two large-$\ell$ cells improve --- a
small legal matter moves from $\xi^{*} = 0.583$ to $0.688$, a per-study radiology read from
$0.573$ to $0.679$, because the equilibrium multiple is $4.00$ rather than the assumed $3$ ---
while the two mid-size cells collapse from $\xi^{*} = 0.167$ and $0.187$ to zero, because their
post-exclusion solvency ratio of $1.5$ lies below the watershed $m_u = 1.77$. The Austrian column
collapses entirely, which is the model's least comfortable prediction and its most testable one:
the moderation right of \S 1336(2) ABGB should, on this account, erode the observed multiple
toward compensation rather than hold it at the $1.2$--$2$ band the companion paper reports. The
German column does not escape. The German map
has a watershed at $1.53$, and the two German mid-size cells, whose post-exclusion solvency ratio
is $1.5$, fall below it by $0.03$ and collapse with the English ones. What \S 348 HGB does is
lower the watershed from $1.77$ to $1.53$: German doctrine survives every shock that keeps the
solvency ratio above $1.53$, and the two large-$\ell$ German cells improve. The doctrinal
channel is not switched off in Germany; it is less fragile there.

\subsection{Robustness}

\begin{table}[t]
\centering\small
\caption{Alternative calibrations. The upper block uses case populations of the kind described
in Section~\ref{sec:calib}; sd is the standard deviation of $m_{\mathrm{crit}}$. The last row
is a logistic law with the population's mean and standard deviation; it reports the logistic
scale in the sd column.}
\label{tab:rob}
\begin{tabular}{lrrrrl}
\toprule
Case & $\lbar$ & $c$ & sd & $\lbar/\lbar^{\dagger}$ & fixed points\\
\midrule
Baseline, $\xi\kbar \in [0.20, 0.60]$ & 3.000 & 2.135 & 0.436 & 1.54 & 1.000, 1.765, 4.000\\
$\xi\kbar \in [0.20,0.40]$, companion headline & 3.000 & 1.772 & 0.158 & 2.74 & 1.000, 1.625, 4.000\\
$\xi\kbar \in [0.25,0.55]$, compressed & 3.000 & 2.098 & 0.318 & 1.73 & 1.000, 1.854, 4.000\\
$\xi\kbar \in [0.15,0.65]$, widened & 3.000 & 2.186 & 0.572 & 1.45 & 1.000, 1.668, 4.000\\
English central multiple 3 as cap & 2.000 & 2.135 & 0.436 & 1.02 & 1.000, 2.200, 2.800\\
Priest--Klein direction & 4.000 & 2.135 & 0.436 & 2.05 & 1.000, 1.684, 5.000\\
\addlinespace
\multicolumn{6}{l}{\emph{logistic benchmark}}\\
Same mean and sd & 3.000 & 2.135 & 0.240 & --- & 1.030, 1.950, 3.999\\
\bottomrule
\end{tabular}
\end{table}

Table~\ref{tab:rob} reports the calibrations we ran. Every population case retains three fixed
points, with compensation at exactly one and the permissive fixed point at the top of the
doctrinal range whenever the population lies below it; the watershed lies between $1.62$ and
$2.20$. The watershed is the quantity that moves, and it is the quantity
Proposition~\ref{prop:collapse} makes the collapse threshold, so the level of the threshold should
be read as a band of roughly $1.6$--$2.2$ rather than as $1.77$. Narrowing the sharing range to
the companion paper's headline interval $\xi \in [0.2, 0.4]$ tightens the population sharply
(sd $0.158$) and lowers the watershed to $1.62$; compressing or widening the range around the same
midpoint moves the watershed by less than $0.1$ and creates no further equilibrium.

\subsection{Measuring the two AI-side parameters in 2026}
\label{sec:measure}

The model carries exactly two parameters that describe the AI rather than the law: the
intensity $\xi\kbar$ with which reviewer and producer draw on the same base model, and the
base provability $\theta_0$ of an error once it has occurred. Both were inherited as
calibration targets. Both now have instruments.

\paragraph{Shared base-model intensity.} That correlated model use is not a modelling
convenience is now documented at scale. \citet{kim2025} evaluate more than three hundred
and fifty language models and report that on one leaderboard task two models agree on the same
wrong answer about sixty per cent of the time when both err, with the correlation surviving
differences in vendor, architecture and size. \citet{ballestero2026} separate baseline
similarity from strategically induced convergence and find agreement of seventy-two per cent
under coordination incentives against thirty-one per cent for human subjects, while divergence
when divergence is rewarded reaches only twenty-seven per cent against three and a half.
Neither study measures $\xi$ for a given pair of firms, but both bound it away from zero for
any pair drawing on the current frontier.

For the pair itself, model fingerprinting has become practical. \citet{wufingerprint2025}
classify fifty-eight models into families from gradient responses to random input
perturbations, without access to training data or watermarks, at ninety-four per cent accuracy.
That is the shape an audit instrument would take: a deployer who claims independent review can
be asked to show that the reviewing model is not of the producing model's family, and the claim
is checkable by a third party. We do not calibrate $\xi$ from fingerprinting here --- no public
dataset pairs firms with their model lineages --- but the parameter has ceased to be
unmeasurable in principle, which is what the falsification conditions of
Section~\ref{sec:falsify} require.

\paragraph{Provability.} $\theta_0$ is the probability that an error, having occurred, can be
established to the standard a court applies. Three developments move it, all of them upward and
all of them dated. The AI Act obliges providers to build automatic logging into high-risk
systems (Article~12) and to retain the resulting logs for at least six months (Article~19),
with the same minimum on the deployer (Article~26(6)); those logs are the evidentiary substrate
a claimant previously lacked. The recast Product Liability Directive shifts the burden where
technical complexity makes proof excessively difficult. And the Article~50 labelling duties,
together with the content-credential infrastructure that major providers extended through 2026,
attach provenance to machine-generated artefacts.

None of these instruments defines
provability as a measured quantity. There is no benchmark, no certified procedure, and no
codified evidentiary standard that would let one read $\theta_0$ off an audit report; the legal
sources establish that proof has become easier, not by how much. The values used in
Section~\ref{sec:calib} remain a calibration on published error-rate and enforcement evidence,
as they were, and the model's predictions are stated as conditional on them. What has changed
since the companion papers is that the direction of movement in $\theta_0$ is now legally
determined rather than assumed, which is enough to sign the comparative statics and not enough
to pin the level.

\paragraph{Provability as a way back.} The direction matters for the trap. A rise in $\theta_0$
lowers every case's threshold, shifts the population down and lowers the watershed; it does not
remove the compensation fixed point, because no case separates at compensation while
$\theta_0 < 1$. Scaling every $\theta_0$ by $1.05$ and $1.10$, and up to the point at which the
best engagement type reaches certainty, lowers the watershed from $1.77$ to $1.64$, $1.54$ and
$1.40$, the barrier out of the compensation well from $0.099$ to $0.072$, $0.051$ and $0.029$,
and the expected recovery time at $\sigma = 0.20$ from $10^{2.5}$ to $10^{1.9}$, $10^{1.4}$ and
$10^{0.9}$ periods. Provability is therefore a way back in the stochastic sense --- it shortens
the wait for a revision large enough to clear the barrier --- and not in the deterministic one:
a market that has converged to compensation stays there under the adjustment of
Proposition~\ref{prop:stab} however far provability rises short of certainty. The order-theoretic
result of Section~\ref{sec:lean} does not cover this path, because a map raised by better
provability is not dominated by the pre-shock map; the statement above is what takes its place.

\section{Scope, robustness and falsification}
\label{sec:falsify}

\paragraph{The recognition premise.}
Assumption~\ref{as:recog} is the load-bearing one and no decided case supports it. We are aware
of no judgment testing whether the maintenance of professional capability behind an AI-assisted
deliverable is a legitimate interest under the \emph{Cavendish} test. The model is conditional:
\emph{if} such an interest is recognised, the doctrine has the fixed-point structure derived
here. A holding that it is not recognised sets $\lbar = 0$ for that jurisdiction and collapses
the model to the American case, which is a clean falsification and one that a single first-instance
decision could deliver.

\paragraph{Selection in the observed multiples.}
$\lbar$ is calibrated from litigated outcomes, and litigated outcomes are selected
\citep{priest1984}. \emph{Houssein} establishes that a fourfold multiple survived review on those
facts; it does not establish that fourfold is the ceiling, and the parties who settled at higher
multiples are invisible. The direction of the bias is determinate --- observed non-penal multiples
understate the enforceable maximum --- so $\lbar = 3$ is conservative and the watershed
$m_u = 1.77$ is, if anything, too high: at $\lbar = 4$ the watershed falls to $1.68$ and the
Maxwell point rises to $1 + \lbar/2 = 3.00$, widening England's margin on both thresholds. That
direction favours the permissive fixed point and therefore cuts against this paper's more
pessimistic results.

\paragraph{Endogenous solvency.}
The shock experiment holds $\ell$ exogenous, and \citet{che2008} is the standing objection:
capital structure is chosen, and a provider who understands Section~\ref{sec:complement} has a
strategic motive to hold $\ell$ \emph{above} $m_u$ --- costly mezzanine capital, cash buffers,
or splitting engagements to shrink $v$. Three points. First, the objection sharpens rather than
blunts the mechanism: if providers park capital just above the watershed, the cross-sectional
distribution of $\ell$ should bunch immediately above $m_u = 1.77$, which is a testable
prediction this model makes and an exogenous-capital model does not. Second, the January 2026
event was a withdrawal of \emph{insurance} capacity, which repriced the cost of holding
$\ell$ discretely and economy-wide; substituting equity for cover is slow and expensive at
exactly the moment the incentive bites, which is why the exogenous treatment is the right first
approximation for that episode. Third, a full treatment makes $\ell$ a reaction function
$\ell = f(m_u)$ and smooths or relocates the jump of Proposition~\ref{prop:collapse}; that
extension is flagged, not solved here. It is not unaddressed in the series, however:
\citet{bauer2026e} makes the ceiling the endogenous object and finds that the share of it
carried by a non-monitoring underwriter is informationally inert at any size, so that the
capital which can actually support a separating pledge is bounded well below the posted limit.
That result runs with the mechanism rather than against it --- $\ell$ is not freely chosen even
once capital structure is endogenous --- but it does not by itself deliver the reaction function
$f$, and we do not claim one.

\paragraph{The primary-obligation carve-out.}
\emph{Cavendish} reaches only secondary obligations. A commitment drafted as a primary obligation
or as a price-adjustment term escapes review entirely, in which case $m = \infty$ and the doctrine
ceiling never binds. This is not a limitation of the model so much as a fourth margin the model
does not price: drafting. If drafting around the doctrine were costless, the doctrine ceiling
would bind nowhere. That it is not is an empirical regularity --- caps and liquidated-damages clauses
remain the dominant form in professional-services contracting --- but the model has nothing to say
about why, and a paper that endogenised the drafting choice would nest this one. The shape of
that paper is already visible. Let $k(\xi)$ be the drafting cost of a primary-obligation
structure delivering the same separation, and let $w(M)$ be the separation value of an
effective ceiling $M$ from the companion paper. The drafting margin switches when
\begin{equation}\label{eq:switch}
w(\ell) - w\big(\min\{\ell, m^{*}\}\big) \;>\; k(\xi),
\end{equation}
and the left side is maximised exactly at the post-shock point, where $m^{*} = m^{*}_{L}$: the
collapse of Section~\ref{sec:complement} is the moment the carve-out is most valuable. Two
consequences follow; the first is a proposition, the second a conjecture.

\begin{proposition}[Drafting cap on the trap]\label{prop:cap}
With free choice of contractual form, the realised separation value at doctrine $m$ is
$\max\{w(\min\{\ell, m\}),\, w(\ell) - k(\xi)\}$, so the loss relative to the
solvency-only benchmark $w(\ell)$ is
\begin{equation}\label{eq:cap}
w(\ell) - \max\big\{w(\min\{\ell, m\}),\; w(\ell) - k(\xi)\big\}
\;=\; \min\big\{\, w(\ell) - w(\min\{\ell, m\}),\;\; k(\xi) \,\big\}.
\end{equation}
The compensation trap costs at most the drafting cost, whatever the doctrine does; and it binds
--- the market stays in secondary form at $m^{*}_{L}$ --- if and only if \eqref{eq:switch}
fails there, that is, $k(\xi) \ge w(\ell) - w(m^{*}_{L})$.
\end{proposition}

The object level matters. The cap acts on payoffs, not on the potential: clipping $U(m)$ at
$k(\xi)$ would conflate the doctrine's state space with the parties' payoff space, and
\eqref{eq:switch}--\eqref{eq:cap} are the correct formalisation of the arbitrage margin. What
remains a conjecture is the feedback channel: clauses that leave the penalty regime leave the
case flow --- in the perturbed dynamics, $dN_t = -\Gamma N_t\,
\mathbf{1}\{\eqref{eq:switch}\ \text{holds}\}\,dt$, stated and not solved --- the
population from which \eqref{eq:D} is formed thins, and a doctrine without cases is frozen at
its last position: an absorbing state coupled to the drafting margin, which connects to the
machine-checked orbit results in the most literal way: no new iterations arrive. Solving the
joint fixed point in $(m, F, N)$, with $F$ the contractual form, is the natural next paper in
the series.

\paragraph{Functional form.}
Multiplicity requires that the doctrine lift the anticipated multiple above itself somewhere,
$\lbar > \lbar^{\dagger}$, and the calibration uses the exact distribution of the case population
rather than a functional form. What the shape of that distribution decides is the number and
location of watersheds: a unimodal density gives one; a bimodal one --- two distinct populations
of cases, say small scoped matters and large projects --- can give two and five fixed points; a
population whose thresholds lie far above the doctrinal range pushes $\lbar^{\dagger}$ beyond any
plausible uplift and removes multiplicity altogether. The tested implication is not a functional
form but the existence of a watershed, and its observable counterpart is set out next.

\paragraph{What would falsify the argument.}
First, a decision rejecting the capability interest, as above. Second, evidence that, across
clause families with comparable primitives in an interest-referenced jurisdiction, the share of
clauses enforced above compensation is unimodal. Multiplicity predicts that families sit in
different basins, so that this share is near zero in some and near one in others; determinacy
predicts a common share. The test is on family-level recognition rates and not on individual
enforced multiples, which Assumption~\ref{as:interest} makes two-point in every regime and which
therefore cannot tell the regimes apart. Third, and specific to
Proposition~\ref{prop:collapse}, evidence that the enforceable multiple in a jurisdiction did
\emph{not} move after the January 2026 exclusions attached; the model predicts a doctrinal
response with a lag in England, and in Germany only in cells whose post-exclusion solvency ratio
fell below $1.53$. Fourth, and most demanding, evidence that
recognition of the interest is unrelated to whether the market separates --- that courts recognise
a capability interest in markets where no commitment could have certified it. That would
invalidate Assumption~\ref{as:recog} directly and leave the doctrine exogenous after all, which
is the companion paper's original treatment. Fifth, the noise
intensity: metastable irreversibility holds on a policy horizon $T$ with at least ninety per cent
probability only while the volatility of anticipated multiples stays below $\sigma^{\ddagger}(T)$
(Corollary~\ref{cor:meta}); measured contract-corpus volatility above
$\sigma^{\ddagger}(30) = 0.20$ on an annual revision clock would make a return of the permissive
doctrine within a generation more than a one-in-ten event, and above $\sigma^{\dagger}(10) = 0.35$
more likely than not within a decade. The parameter is measurable in principle from the dispersion
of year-on-year changes in agreed multiples within a stable clause family, and we commit to that
test rather than treat $\sigma$ as free. Sixth, drafting migration: a post-shock shift of
professional-services contracting toward primary-obligation structures --- price-adjustment
terms, conditional fee architectures --- at the scale of the market would show providers
exiting the doctrine rather than being trapped by it, would activate \eqref{eq:switch}, and
would cap the welfare cost of the trap at the drafting cost $k(\xi)$; persistent dominance of
liquidated-damages form after the shock is the observable that keeps the trap reading alive.

\paragraph{Predictions.}
To make the conditional structure bite, we state the tests before any court has ruled on the
capability interest. (P1) In
England, year-on-year volatility of agreed multiples within a stable clause family will remain
below $\sigma^{\ddagger}(30) = 0.20$ on an annual revision clock; measured volatility above that
level makes a return of the permissive doctrine within a generation more than a one-in-ten event,
and above $\sigma^{\dagger}(10) = 0.35$ more likely than not within a decade. (P2) Across clause
families with comparable primitives in interest-referenced jurisdictions, the share of clauses
enforced above compensation will be bimodal --- near zero in some families and near one in
others, with few in between --- the observable counterpart of the watershed; in the United States
and Austria it will be unimodal near zero. (P3) If cover is not restored while mean model sharing
continues to rise, drafting migration toward primary-obligation structures will begin as the
permissive doctrine becomes likely to be lost on the parties' contracting horizon: at
$c^{\dagger}(\sigma, T)$, which on the monoculture path of Section~\ref{sec:stoch} lies at mean
sharing $0.51$--$0.52$, just before the fold at $0.52$, and not at the Maxwell point at $0.48$,
which marks only where the very-long-run selection flips. (P4) The cross-sectional distribution of
solvency cover $\ell$ among providers who keep the secondary form will bunch immediately above the
watershed, $m_u = 1.77$ in England and $1.53$ in Germany. Each prediction names its data source in
the text where it is derived; jointly they replace an untested legal assumption with a dated
forecast.

\paragraph{Inherited limitations.}
Everything the companion paper cannot do, this paper cannot do either: $\xi$ is not observable
from provenance disclosures, $\kappa_0$ is an agreement-on-errors anchor rather than an error
decomposition, and the survival condition is a capability threshold rather than a calendar
prediction. Two extensions identified in that paper remain open here and interact with this one:
endogenous model sharing, under which $\xi$ is chosen in a diversification game and $c$ therefore
moves with the equilibrium of that game, and the provability floor as a fixed point in $L$.
Neither is solved here.

\section{Policy}

The model has four implications for policy.

\emph{Depth and duration of a capacity shock.} By Corollary~\ref{cor:irrev}, a capital
constraint that binds deeply enough and long enough to move the doctrine across its watershed is
not undone when the constraint is relaxed. The two conditions are of different kind. A withdrawal
that keeps the solvency ratio above the watershed leaves no lasting effect at any duration; one
that takes it below has a lasting effect after a fuse of two periods for a deep shock, and after a
longer one the closer the shock stays to the watershed. If the mechanism operates, the relevant
question during the current insurance withdrawal is whether capacity falls below the watershed.
Interventions that keep it above --- a backstop, a mandatory minimum limit, a pooled facility with
provider-specific differentiation --- are then worth more, even at a modest level, than a larger
capacity restored later. The companion paper's fourth falsification condition, a deep affirmative market,
is necessary but not sufficient.

\emph{Thresholds and horizons.} Deterministic comparative statics watch for the fold at
$\bar c = 3.05$, where the permissive equilibrium ceases to exist. Section~\ref{sec:stoch} adds two
thresholds, each with its own time scale. The Maxwell point $c^{\ast} = 1 + \lbar/2 = 2.50$ --- mean
model sharing of about $0.48$ --- is where the very-long-run selection flips; at the noise levels at
which the collapse is irreversible, a permissive doctrine there survives thousands of periods or
more. The policy-horizon threshold $c^{\dagger}(\sigma, T)$ is where a permissive doctrine becomes
likely to be lost within $T$; on a thirty-period horizon it lies at $2.83$--$3.01$, mean sharing
$0.51$--$0.52$, just short of the fold. A regulator watching the fold is watching roughly the right
threshold for the observable collapse. The Maxwell point carries different information:
beyond it the compensation doctrine is the deeper well, so a collapse, once it happens, is harder
to undo by chance than the permissive doctrine is to lose.

\emph{Doctrinal clarification.} In the indeterminate band the equilibrium is
selected by expectations, so an authoritative statement --- appellate guidance, a restatement, a
model-clause endorsement by a professional body --- that the capability interest is legitimate
selects the permissive fixed point at no fiscal cost. Symmetrically, prolonged doctrinal silence
during a period when few cases separate selects the compensation fixed point by default. Nothing
in the model requires the announcement to change any court's behaviour in any decided case; it
changes the expectation that enters the drafting stage, which is what \eqref{eq:D} iterates. In
the language of Section~\ref{sec:stoch}, such an announcement is a jump that clears the barrier
in one step, not a diffusion across it.

\emph{Regime-specific instruments.} The companion
paper concludes that relaxing the enforceable multiple does nothing in a solvency-bound cell and
that expanding capacity does nothing in a doctrine-bound one. Proposition~\ref{prop:collapse}
sharpens this: because $\ell$ enters the doctrine map, a cell can be solvency-bound today and
doctrine-bound tomorrow \emph{as a consequence of the same shock}, so an instrument calibrated to
the observed regime addresses the constraint that binds now and not the one that will bind after
the doctrine adjusts.

\section{Conclusion}

The companion paper treats the penalty doctrine as one of three exogenous ceilings on a liability
signal and reads a regime map off the minimum of two of them. This paper takes the doctrine's own
test seriously. Post-\emph{Cavendish} the enforceable multiple is measured against a legitimate
interest, and where that interest is the capability the commitment certifies, the multiple is a
fixed point rather than a datum. The map has the structure that follows: a fixed point always
exists, compensation is always one of them because no case separates there, and uniqueness fails
once the doctrinal uplift clears a threshold set by the case population. The resulting watershed
sorts jurisdictions into a determinate compensation end --- the United States and Austria --- and
an indeterminate remainder in which England, India since December 2025, and negotiated German
contracting sit, the last with a much smaller compensation basin.

The result we did not expect is the interaction. Because separation requires the smaller of the
two ceilings, a court asked to recognise the capability interest in a market where capital binds
sees nothing to protect, and the doctrine map inherits the capital constraint. The two ceilings
are therefore complements, not arguments of a minimum, and the equilibrium ceiling jumps
downward when the solvency ratio crosses the doctrinal watershed. The January 2026 insurance
exclusions are, on this reading, not only a civil-law event that moves cells across a boundary.
In every cell whose solvency ratio falls below the watershed --- common-law cells, and German
cells below $1.53$ --- they are a shock that can destroy a legal equilibrium, and once the shock
has outlasted its fuse the destruction does not reverse when the cover returns.

\section{Machine-checked results (Lean 4)}
\label{sec:lean}

The order-theoretic layer of this paper is machine-checked. The file
\texttt{lean/DoctrineOrder.lean} is self-contained Lean~4 (toolchain 4.15.0): it uses only Lean's
core library --- no mathlib, no Batteries --- defines a bundled complete-lattice structure, and
proves seventeen theorems, the nine results listed below and eight auxiliary lemmas, with no
proof step left open. For every one of the seventeen the checker reports \emph{``does not
depend on any axioms''}: the proofs are fully constructive and use neither propositional
extensionality nor the axiom of choice.

The nine results, with their counterparts in the text:
\texttt{lfp\_fixed} and \texttt{gfp\_fixed} (Knaster--Tarski, existence of least and greatest
fixed points of a monotone self-map --- Proposition~\ref{prop:exist});
\texttt{lfp\_least} and \texttt{gfp\_greatest} (extremality);
\texttt{watershed\_no\_crossing} (orbits starting weakly below a fixed point never cross it ---
the order-theoretic core of Corollary~\ref{cor:irrev});
\texttt{dominated\_orbit} (pointwise-dominated maps have dominated orbits --- the comparison
between the constrained map $D(\cdot\,;\ell)$ and $D$);
\texttt{irreversibility} (under \emph{any} sequence of doctrine maps each pointwise dominated by
the recovered map --- shock, partial recovery, full recovery --- a state once below the
watershed remains below it at every future date); and
\texttt{kleene\_fixed} with \texttt{orbitSup\_le\_fixed} ($\omega$-continuous convergence of
orbits to the least fixed point above the start). The monotone-orbit structure behind
Proposition~\ref{prop:stab} rests on the auxiliary lemmas \texttt{iter\_le\_fixed},
\texttt{fixed\_le\_iter}, \texttt{iter\_ascending} and \texttt{iter\_descending}.

Applicability to the adjustment dynamics is immediate: the partial-adjustment map
$T(m) = m + \eta(D(m) - m)$ is nondecreasing because $D$ is and $\eta \le 1$; its fixed
points are exactly those of $D$; and the shocked map
$T_{\ell}(m) = m + \eta(D(m;\ell) - m)$ is pointwise dominated by $T$ because
$D(\cdot\,;\ell) \le D$. The hypotheses of \texttt{irreversibility} are therefore satisfied by
the objects of Section~\ref{sec:complement} directly. They are not satisfied by a path along which
provability rises: a better-proved market has a pointwise higher map, which is not dominated by the
pre-shock one. Section~\ref{sec:measure} states what holds on that path instead.

The analytic results --- the multiplicity
threshold $\lbar^{\dagger}$, the saddle-node folds, the Maxwell closed form
$c^{\ast} = 1 + \lbar/2$, and the Kramers times --- are \emph{not} formalised. They involve
transcendental inequalities for which bare Lean has no real-analysis library; formalising them
against mathlib's \texttt{Analysis} hierarchy is a genuine but separate project. What is
machine-checked is exactly the layer on which the paper's irreversibility logic rests.

\appendix

\section{Proofs}
\label{app:proofs}

\begin{proof}[Proof of Proposition~\ref{prop:map}]
Fix the anticipated multiple $m$. By Assumptions~\ref{as:interest} and~\ref{as:recog} a case
with threshold $m_{\mathrm{crit}}$ is enforced at $\bar m = 1 + \lbar$ if
$m_{\mathrm{crit}} \le m$ and at $1$ otherwise; by Assumption~\ref{as:hetero} the first event has
probability $F(m)$. The multiple the market anticipates for the next period is the expected
enforced multiple, $1 \cdot (1 - F(m)) + (1 + \lbar) F(m) = 1 + \lbar F(m)$, which is
\eqref{eq:D}. $F$ is a continuous distribution function, so $D$ is continuous, nondecreasing and
takes values in $[1, 1+\lbar]$. If the support of $m_{\mathrm{crit}}$ is an interval, $F$ is
strictly increasing on it, and so is $D$ when $\lbar > 0$; in the calibration the four type
supports $[1/(\theta_0 \cdot 0.8),\, 1/(\theta_0 \cdot 0.4)]$ overlap and their union is
$[1.4706, 3.2051]$. For $\lbar = 0$, $D \equiv 1$. Where $F$ has a density, $D' = \lbar f$.
Finally, Definition~\ref{def:mcrit} gives $m_{\mathrm{crit}} \ge 1/\theta_0 > 1$ in every case,
so $\underline m > 1$, $F \equiv 0$ on $[1, \underline m]$, and $D \equiv 1$ there; in particular
$D(1) = 1$. The remark after the proposition follows the same way: a clause enforced at the
anticipated $m$ when recognised gives $1 + (m-1)F(m) \le m$, with equality exactly where
$(m - 1)(1 - F(m)) = 0$: at $m = 1$ and wherever $F(m) = 1$. Below the upper end of the support the
inequality is strict, so orbits starting there decrease to the only fixed point below it, $m = 1$.
\end{proof}

\begin{proof}[Proof of Proposition~\ref{prop:exist}]
$[1, 1+\lbar]$ with the usual order is a complete lattice, and $D$ is monotone by
Proposition~\ref{prop:map} and maps the lattice into itself. Tarski's theorem
\citep{tarski1955} gives that the set of fixed points is a non-empty complete lattice.
\end{proof}

\begin{proof}[Proof of Proposition~\ref{prop:mult}]
Write $G = D - \mathrm{id}$. On $(1, \underline m]$, $G(m) = 1 - m < 0$, and for
$m > \underline m$, $F(m) > 0$ and
\[
G(m) > 0 \iff \lbar > \frac{m-1}{F(m)} .
\]
The ratio $(m-1)/F(m)$ is continuous on $(\underline m, \infty)$ and tends to $+\infty$ at both
ends, so its infimum $\lbar^{\dagger}$ is attained, at $m^{\dagger}$ say.

(i) If $\lbar < \lbar^{\dagger}$, then $G < 0$ on $(1, 1+\lbar]$ and $G(1) = 0$, so compensation
is the only fixed point. From any $m_0 > 1$ the adjustment of Proposition~\ref{prop:stab} gives
$m_{t+1} - m_t = \eta G(m_t) < 0$ with $m_{t+1} \ge 1$, so $m_t$ decreases to a fixed point,
which can only be $1$. For $\lbar = 0$ this is the constant map.

(ii) If $\lbar > \lbar^{\dagger}$, then $G(m^{\dagger}) > 0$, which forces
$m^{\dagger} < 1 + \lbar$ because $D \le 1 + \lbar$. Since $G < 0$ on $(1, \underline m]$, the
intermediate value theorem gives a zero $m_u \in (\underline m, m^{\dagger})$, the least fixed
point above $1$. Since $G(1+\lbar) = \lbar\,(F(1+\lbar) - 1) \le 0$, it gives a zero in
$(m^{\dagger}, 1+\lbar]$, and the greatest fixed point $m^{*}_{H}$ lies there. Together with
$D(1) = 1$ these are three distinct fixed points with $\underline m < m_u < m^{*}_{H}$. If the
support lies below $1 + \lbar$, then $F(1+\lbar) = 1$, $G(1+\lbar) = 0$, and the greatest fixed
point is $1 + \lbar$. At $\lbar = \lbar^{\dagger}$, $G \le 0$ on $(1, 1+\lbar]$ with equality
exactly on the set of minimisers of the ratio, where the watershed and the permissive fixed point
coincide.

(iii) If the density is strictly unimodal, $G' = \lbar f - 1$ is negative, then positive, then
negative on three consecutive (possibly empty) intervals, and nowhere zero on an interval. $G$ is
therefore strictly monotone on each, with at most one zero on each. On the first, which contains
$[1, \underline m]$, the zero is $m = 1$. Hence there are at most three fixed points, and by (ii)
exactly three.
\end{proof}

\begin{proof}[Proof of Proposition~\ref{prop:stab}]
The adjustment map is $T(m) = m + \eta(D(m) - m)$. Because $D$ is nondecreasing and
$\eta \le 1$, $T$ is nondecreasing, so every orbit is monotone and bounded and converges to a
fixed point. Where $D$ is differentiable at $m^{*}$, $T'(m^{*}) = 1 - \eta(1 - D'(m^{*}))$;
since $D' \ge 0$ and $\eta \le 1$, $T'(m^{*}) \ge 1 - \eta > -1$, so $|T'(m^{*})| < 1$ if and
only if $D'(m^{*}) < 1$, and $T'(m^{*}) > 1$ if $D'(m^{*}) > 1$. At the ends of the range the
same holds with the one-sided derivative; at $m = 1$ it is $0$, because $D$ is flat on
$[1, \underline m]$. At a kink of $F$ in the interior the fixed point is stable if both one-sided
derivatives are below one and unstable if both exceed it. With three fixed points, $G < 0$ on
$(1, m_u)$, $G > 0$ on $(m_u, m^{*}_{H})$ and $G < 0$ on $(m^{*}_{H}, 1+\lbar]$. On $[1, m_u)$,
$T(m) < m$ and $T(m) \ge T(1) = 1$, so orbits decrease to $1$. On $(m_u, m^{*}_{H})$,
$m < T(m) \le T(m^{*}_{H}) = m^{*}_{H}$, so orbits increase to $m^{*}_{H}$. Above
$m^{*}_{H}$ they decrease to it. This gives the basins and shows that the outer fixed points
attract and the middle one repels, whatever the derivatives.
\end{proof}

\begin{proof}[Proof of Corollary~\ref{cor:feas}]
By \eqref{eq:survival}, feasibility requires $\xi\kbar \le 1 - 1/(M\theta_0)$. With
$M = m^{*}_{L} = 1$ the right-hand side is $1 - 1/\theta_0$, which is strictly negative for every
$\theta_0 < 1$, and no $\xi \ge 0$ satisfies the inequality. At $\theta_0 = 1$ the right-hand
side is zero and $\xi = 0$ satisfies it, which is why the statement excludes perfect
provability.
\end{proof}

\begin{proof}[Proof of Proposition~\ref{prop:fold}]
Write $s = c - c_0$, let $F_0$ be the distribution function at $c_0$ with support
$[\underline m_0, \bar m_0]$, and let $g_0(u) = 1 + \lbar F_0(u) - u$, which does not depend on
$c$. Because $F(m; c) = F_0(m - s)$,
\[
G(m; c) \;=\; g_0(m - s) - s,
\]
so $m$ is a fixed point at $c$ if and only if $u = m - s$ solves $g_0(u) = s$. Every solution is
admissible: below the support $g_0(u) = 1 - u$, whose only solution is $m = 1$; above it
$g_0(u) = 1 + \lbar - u$, whose only solution is $m = 1 + \lbar$; inside it $0 \le F_0 \le 1$
places $m$ in $[1, 1+\lbar]$. The fixed points at $c$ are the level set $\{g_0 = s\}$ shifted by
$s$. $g_0$ is continuous and tends to $+\infty$ as $u \to -\infty$ and to $-\infty$ as
$u \to +\infty$. Let $g^{+} = \max_{u \ge \underline m_0} g_0(u)$, attained at $u^{+}$, and
$g^{-} = \min_{\underline m_0 \le u \le u^{+}} g_0(u)$, attained at $u^{-}$. Because
$\lbar > \lbar^{\dagger}$, $g_0(m^{\dagger}) = G(m^{\dagger}; c_0) > 0$, so $g^{+} > 0$; because
$\underline m_0 > 1$, $g^{-} \le g_0(\underline m_0) = 1 - \underline m_0 < 0$. Set
$\underline c = c_0 + g^{-}$ and $\bar c = c_0 + g^{+}$, so $\underline c < c_0 < \bar c$.

For $s \in (g^{-}, g^{+})$ the intermediate value theorem gives a solution in each of
$(-\infty, u^{-})$, $(u^{-}, u^{+})$ and $(u^{+}, \infty)$: at least three fixed points. If the
density is strictly unimodal, $g_0' = \lbar f_0 - 1$ changes sign at most twice, from negative
to positive to negative, so $u^{-}$ and $u^{+}$ are the only local extrema of $g_0$, and
$\{g_0 = s\}$ has three points for $s \in (g^{-}, g^{+})$, two for $s \in \{g^{-}, g^{+}\}$ and
one otherwise. At $s = g^{+}$ the upper two solutions merge at $u^{+}$ and vanish. For
$c > \bar c$ the only fixed point is compensation, which attracts every orbit by
Proposition~\ref{prop:stab}; if $c$ is then lowered slowly, the state follows the lower branch,
which is stable and persists until $c = \underline c$, so the upper branch is not recovered while
$c > \underline c$. None of this uses differentiability of $F$.

In the calibration ($\lbar = 3$) the maximum is $g^{+} = 0.91720$ at $u^{+} = 2.9412$, where the
first engagement type leaves the support, and the minimum is $g^{-} = -0.47145$ at
$u^{-} = 1.5244$, where the second enters; hence $\underline c = 1.6635$ and $\bar c = 3.0521$.
Both extrema sit at kinks of $F$, so the two merging fixed points meet at a corner of $D$, where
$D'$ jumps across one, not at a tangency. Beyond $u^{+}$, $g_0$ is strictly decreasing, so for
$s < g^{-}$ and for $s > g^{+}$ the level set is a single point. Before $u^{-}$ it is not
monotone: $g_0$ falls to $1 - \underline m_0 = -0.470588$ at $\underline m_0 = 1.4706$, rises
while $\lbar f_0 > 1$ to $-0.470443$ at $u = 1.4852$, and falls again to $g^{-}$. For $s$
between those two values, $c \in (1.66432, 1.66447)$, the level set has five points, and the
additional stable and unstable fixed points lie in $(1, 1.03)$. A direct count on a grid in $c$
confirms one fixed point outside $[\underline c, \bar c]$, three inside it and five in the
window, and the gap $m^{*}_{H} - m^{*}_{L} = 3.00$ at $c = 2.5$.
\end{proof}

\begin{proof}[Proof of Proposition~\ref{prop:constrained}]
Immediate from Assumption~\ref{as:recog} with $m_{\mathrm{crit}}$ compared against
$M = \min\{\ell, m\}$ rather than $m$: the recognition probability is
$\Pr\{m_{\mathrm{crit}} \le \min\{\ell, m\}\}$.
\end{proof}

\begin{proof}[Proof of Proposition~\ref{prop:collapse}]
Fixed points of $D(\cdot;\ell)$ are of two kinds. (i) $m < \ell$: then
$D(m;\ell) = D(m)$, so $m$ is a fixed point of the unconstrained map lying below $\ell$.
(ii) $m \ge \ell$: then $D(\cdot;\ell)$ is constant at $D(\ell)$, so $m = D(\ell)$ is a fixed
point provided $D(\ell) \ge \ell$, i.e.\ $G(\ell) \ge 0$. With exactly three fixed points,
$G(1) = 0$, $G < 0$ on $(1, m_u)$, $G \ge 0$ on $[m_u, m^{*}_{H}]$ and $G < 0$ above
$m^{*}_{H}$ (proof of Proposition~\ref{prop:stab}); for $\ell < 1$, $G(\ell) = 1 - \ell > 0$.

Take $\ell \ge m^{*}_{H}$. Then $m^{*}_{H} \le \ell$ is the greatest fixed point and the limit of
the path from the upper basin, so $M = \min\{\ell, m^{*}_{H}\} = m^{*}_{H}$.

Take $m_u \le \ell < m^{*}_{H}$. Kind (ii) applies since $G(\ell) \ge 0$: from above, the path
converges to $D(\ell) \ge \ell$ without falling below $\ell$, hence $M = \ell$.

Take $m^{*}_{L} < \ell < m_u$. Then $G(\ell) < 0$, so no fixed point of kind (ii) exists. From
the upper basin $m_t$ moves toward $D(\ell) < \ell$ while $m_t \ge \ell$, so it falls below
$\ell < m_u$ in finite time; thereafter the unconstrained dynamics apply from a point in the basin
of $m^{*}_{L}$, and $M = m^{*}_{L}$.

Take $\ell \le m^{*}_{L} = 1$. Then $G(\ell) \ge 0$ and kind (ii) gives $m^{*} = D(\ell) = 1$, so
$M = \min\{\ell, 1\} = \ell$.

Discontinuity: as $\ell \downarrow m_u$, $M \to m_u$ along the second branch, while for $\ell$
just below $m_u$, $M = m^{*}_{L}$; the jump is $m_u - m^{*}_{L} > 0$.
\end{proof}

\begin{proof}[Proof of Corollary~\ref{cor:irrev}]
Start at $m_0 = m^{*}_{H}$ with $\ell = \ell_1 < m_u$. While $m_t \ge \ell_1$ the constrained map
is flat at $D(\ell_1)$, and $D(\ell_1) < m_u$: for $\ell_1 \in (1, m_u)$ because $G(\ell_1) < 0$
gives $D(\ell_1) < \ell_1$, for $\ell_1 \le 1$ because $D(\ell_1) = 1$. Hence
$m_t - D(\ell_1) = (1 - \eta)^{t}\,(m^{*}_{H} - D(\ell_1))$ as long as $m_t \ge \ell_1$, and the
path crosses $m_u$ before it reaches $\ell_1$. The first $t$ with $m_t < m_u$ is the smallest
integer with $(1-\eta)^{t} < (m_u - D(\ell_1))/(m^{*}_{H} - D(\ell_1))$, which is
\eqref{eq:tstar} for $\eta < 1$; for $\eta = 1$, $t^{*} = 1$. Once below $m_u$ the state keeps
falling during the shock, because $D(\min\{\ell_1, m\}) \le D(m) < m$ on $(1, m_u)$.

If the shock lasts $s \ge t^{*}$ periods, $m_s < m_u$ when $\ell$ returns to
$\ell_0 > m^{*}_{H}$; the constraint is then slack, $m_s$ lies in the basin of $m^{*}_{L}$ by
Proposition~\ref{prop:stab}, and $m_t \to m^{*}_{L}$. If $s < t^{*}$, then $m_s \ge m_u$, with
equality only when the ratio in \eqref{eq:tstar} is an integer power of $1 - \eta$; for
$m_s > m_u$ the state lies in the upper basin and returns to $m^{*}_{H}$. As $\ell_1 \uparrow m_u$,
$D(\ell_1) \to D(m_u) = m_u$, the ratio in \eqref{eq:tstar} tends to zero and $t^{*}(\ell_1)$ grows
without bound. If $\ell_1 \ge m_u$,
then $G(\ell_1) \ge 0$, the path converges from above to $D(\ell_1) \ge \ell_1 \ge m_u$ and never
enters the lower basin, so the shock is reversed at any duration.

In the calibration, $D(1.5) = 1.0294$ and $(m_u - D(1.5))/(m^{*}_{H} - D(1.5)) = 0.2477$, so
$t^{*} = \lfloor 2.01 \rfloor + 1 = 3$; after one period the state is
$1.0294 + 0.5 \times 2.9706 = 2.5147 > m_u$.
\end{proof}

\subsection*{Proofs for Section~\ref{sec:stoch}}

\paragraph{Proposition~\ref{prop:potential}.}
The drift of \eqref{eq:sde} is $\eta(D(m)-m) = -U'(m)$ by \eqref{eq:potential}, so the process
is a gradient diffusion on a compact interval with reflecting boundaries. The stationary
Fokker--Planck equation $\tfrac{\sigma^{2}}{2}\pi'' + (U'\pi)' = 0$ with zero-flux boundary
conditions integrates once to $\tfrac{\sigma^{2}}{2}\pi' + U'\pi = 0$, whose normalised solution
is $\pi_{\sigma} \propto e^{-2U/\sigma^{2}}$; uniqueness follows from ellipticity on a compact
state space. $U' = -\eta G$ vanishes exactly at fixed points and changes sign from negative to
positive at the stable ones and from positive to negative at the unstable ones (proof of
Proposition~\ref{prop:stab}), which makes the former local minima and the latter local maxima.
Where $D$ is differentiable, $U'' = \eta(1 - D')$. \hfill$\square$

\paragraph{Theorem~\ref{thm:maxwell}.}
(i) Laplace's principle on $\pi_{\sigma} \propto e^{-2U/\sigma^{2}}$ concentrates all mass, as
$\sigma \to 0$, on the global minimiser of $U$; with two candidate wells the comparison is
$U(m^{*}_{H})$ against $U(m^{*}_{L})$, and the identity in \eqref{eq:area} is
\eqref{eq:potential} evaluated between them, so $U(m^{*}_{H}) < U(m^{*}_{L}) \iff A > 0$.

(ii) At the lower fold the compensation well merges with the watershed, so
$U(m^{*}_{L}) = U(m_u) > U(m^{*}_{H})$ and $A > 0$; at the upper fold
$U(m^{*}_{H}) = U(m_u) > U(m^{*}_{L})$ and $A < 0$. $A$ is continuous wherever the outer fixed
points are, so it has a zero between the folds. Where the outer fixed points move continuously
they are locally Lipschitz in $c$ (zeros of $g_0(\cdot) - s$ at which $g_0$ has nonzero
one-sided slopes, in the notation of the proof of Proposition~\ref{prop:fold}), and the Leibniz
rule applies; its boundary terms carry the factor $G = 0$, and $\partial_c D(u; c) = -\lbar f$,
so
\[
A'(c) \;=\; -\lbar \int_{m^{*}_{L}}^{m^{*}_{H}} f(u; c)\, du
\;=\; -\big(D(m^{*}_{H}) - D(m^{*}_{L})\big) \;=\; -\big(m^{*}_{H} - m^{*}_{L}\big) \;<\; 0 .
\]
Where a new outermost fixed point is born, $A$ jumps by the area of the new lobe. In the
calibration this happens once on the three-fixed-point interval, at $c = 1.66432$, where the
lower end of the population reaches one and compensation returns as the least fixed point; the
jump is $+2.8 \times 10^{-6}$, and $A$ is positive below $c = 2.5$ and negative above it on a grid
of $460$ points across $(\underline c, \bar c)$, so $c^{\ast}$ is unique.

(iii) is proved in the text. It uses the support condition only: with the support inside
$(1, 1+\lbar)$, $D \equiv 1$ below it and $D \equiv 1+\lbar$ above it, so the outer fixed points
are the ends of the range, and $\int_{1}^{1+\lbar} F = 1 + \lbar - c$ because
$\int_{1}^{1+\lbar} (1 - F) = \E[m_{\mathrm{crit}}] - 1$ for a distribution on that range.

(iv) If $S(z) + S(-z) = 1$ and $D(m) = 1 + \lbar S((m - c)/\varsigma)$, then at
$c = 1 + \lbar/2$, $G(c+t) + G(c-t) = 2 + \lbar - 2c = 0$: the drift is antisymmetric about the
midpoint of the range. Hence $m$ is a fixed point if and only if $2c - m$ is, the outer fixed
points are mirror images, and $A = \int_{c-t^{*}}^{c+t^{*}} G = 0$. The logistic density is
strictly unimodal, so by (ii) the zero is unique. \hfill$\square$

\paragraph{Proposition~\ref{prop:cdagger}.}
Let $\tau_H$ be the first passage from $m^{*}_{H}$ to $m_u$, with reflection at $1 + \lbar$. The
generator argument of Proposition~\ref{prop:mfpt}, run on $[m_u, 1+\lbar]$, gives
\[
\E[\tau_H] \;=\; \frac{2}{\sigma^{2}} \int_{m_u}^{m^{*}_{H}} \int_{y}^{1+\lbar}
\exp\!\Big(\frac{2\eta}{\sigma^{2}} \int_{y}^{z} G(v; c)\, dv\Big)\, dz\, dy .
\]
$\partial_c G = -\lbar f \le 0$, so the integrand is nonincreasing in $c$. The watershed rises
with $c$ --- it is $u + s$ with $g_0(u) = s$ on a rising branch of $g_0$ --- and $m^{*}_{H}$ does
not rise, so the domain of integration shrinks strictly. Hence $\E[\tau_H]$ is continuous and
strictly decreasing in $c$ on $(c^{\ast}, \bar c)$, and it tends to zero at $\bar c$, where
$m_u$ and $m^{*}_{H}$ merge. If $\E[\tau_H] > T$ at $c^{\ast}$, there is exactly one
$c^{\dagger}(\sigma, T) \in (c^{\ast}, \bar c)$ at which it equals $T$. The barrier is
$\Delta U_H = U(m_u) - U(m^{*}_{H}) = \eta \int_{m_u}^{m^{*}_{H}} G\, dv$, and the Leibniz rule
gives, as in Theorem~\ref{thm:maxwell}(ii), $d\Delta U_H/dc = -\eta\,(m^{*}_{H} - m_u) < 0$. For
fixed $c < \bar c$, $\Delta U_H > 0$ and Laplace's principle gives $\E[\tau_H] \to \infty$ as
$\sigma \to 0$, so $c^{\dagger}(\sigma, T) \to \bar c$. \hfill$\square$

\paragraph{Theorem~\ref{thm:formfree}.}
The argument for (i) of Theorem~\ref{thm:maxwell} used only that the drift is minus the
derivative of a potential --- automatic in one dimension for any continuous drift --- and
Laplace's principle, which concentrates the invariant density on the global minimiser of $U$
when that minimiser is unique and the minima are nondegenerate. If several wells attain the
global minimum, the limit mass is split between them in proportion to their Laplace weights,
$1/\sqrt{U''}$ at an interior well and half of that at a well on the boundary.
\hfill$\square$

\paragraph{Proposition~\ref{prop:mfpt}.}
Let $T(x) = \E_x[\tau]$. The generator identity
$\tfrac{\sigma^{2}}{2}T'' - U'T' = -1$ with $T(m_u) = 0$ and $T'(1) = 0$ (reflection)
integrates by the factor $e^{-2U/\sigma^{2}}$ to
$T'(y) = -\tfrac{2}{\sigma^{2}} e^{2U(y)/\sigma^{2}} \int_1^y e^{-2U/\sigma^{2}}$, and a second
integration from $m^{*}_{L}$ to $m_u$ gives \eqref{eq:mfpt}. For the asymptotics both Laplace
integrals are half-Gaussians. The inner integral has its minimum at $m^{*}_{L} = 1$, the
reflecting boundary, where $U'(1) = 0$ and $U''(1) = \eta$ because $D$ is flat on
$[1, \underline m]$. The outer integral has its maximum at its end point $m_u$. Hence
\[
\frac{2}{\sigma^{2}} \cdot
\tfrac12\sqrt{\frac{\pi\sigma^{2}}{U''(1)}}\,e^{-2U(1)/\sigma^{2}} \cdot
\tfrac12\sqrt{\frac{\pi\sigma^{2}}{\lvert U''(m_u)\rvert}}\,e^{2U(m_u)/\sigma^{2}}
\;=\; \frac{\pi}{2\sqrt{U''(1)\,\lvert U''(m_u)\rvert}}\; e^{2\Delta U_{L}/\sigma^{2}},
\]
the stated form; an interior well would give a full Gaussian and twice the constant. Passage to
the other well doubles the constant again (Eyring--Kramers), leaving the exponential order
unchanged. \hfill$\square$

\paragraph{Corollary~\ref{cor:meta}.}
Write $\beta = 2/\sigma^{2}$. By \eqref{eq:mfpt},
$\E[\tau] = \beta \int_{1}^{m_u} \int_{1}^{y} e^{\beta\,(U(y) - U(z))}\, dz\, dy$. On $(1, m_u)$,
$G < 0$, so $U$ is increasing on $[1, m_u]$ and $U(y) - U(z) \ge 0$ for $z \le y$. The integrand
$\beta\, e^{\beta(U(y) - U(z))}$ is therefore strictly increasing in $\beta$, and $\E[\tau]$ is
strictly decreasing in $\sigma$; it diverges as $\sigma \to 0$, because $\Delta U_{L} > 0$, and
vanishes as $\sigma \to \infty$. Hence $\sigma^{\dagger}(T)$ and $\sigma^{\ddagger}(T)$ exist and
are unique for every $T > 0$, and the reported values invert the exact quadrature by bisection.
The probability reading uses the exponential law of the passage time, $\Pr\{\tau \le T\} =
1 - e^{-T/\E[\tau]}$, which is at most one tenth if and only if
$\E[\tau] \ge T / \ln(10/9) = 9.49\,T$; the Euler--Maruyama check in the text confirms the law
at the three noise levels simulated. \hfill$\square$

\paragraph{Proposition~\ref{prop:cap}.}
Write $a = w(\ell)$, $b = w(\min\{\ell, m\})$ and $k = k(\xi)$, with $b \le a$ by
monotonicity of $w$ and $k \ge 0$. If $a - k \ge b$ the maximum in \eqref{eq:cap} is
$a - k$ and the loss is $k$, which equals $\min\{a - b, k\}$ because $a - b \ge k$;
otherwise the maximum is $b$ and the loss is $a - b < k$. Binding of the secondary form means
the maximum is attained at $b$, i.e.\ $b \ge a - k$, i.e.\ $k \ge a - b$, which at
$m = m^{*}_{L}$ is exactly the negation of \eqref{eq:switch}. \hfill$\square$

\section{Parameter provenance and notation}
\label{app:prov}

\begin{table}[ht]
\centering\small
\caption{Provenance of every parameter. ``Inherited'' means taken from a companion paper without
re-derivation; ``derived'' means computed here from inherited primitives.}
\begin{tabular}{@{}llll@{}}
\toprule
Parameter & Value & Source & Status\\
\midrule
$\theta_0$ by type & 0.80, 0.82, 0.85, 0.78 & Bauer (2026d), App.~B & inherited\\
$\xi\kbar$ range & $[0.20, 0.60]$ & Bauer (2026d), \S4.1, \S6 & inherited, widened\\
$c$ & 2.1349 & $\E[m_{\mathrm{crit}}]$ on that grid, closed form & derived\\
support & $[1.4706, 3.2051]$ & Definition~\ref{def:mcrit} on that grid & derived\\
sd$(m_{\mathrm{crit}})$ & 0.4361 & closed form & derived\\
$\lbar$ (US) & 0 & Restatement \S 356 & primary source\\
$\lbar$ (AT) & 1 & \S 1336(2) ABGB & primary; contested\\
$\lbar$ (UK) & 3 & \emph{Houssein} $4\times$ rate & primary; selected\\
$\lbar$ (IN) & 3 & \emph{BPL} 2025 INSC 1380 & UK value assigned\\
$\lbar$ (DE) & 24 & \S 348 HGB & negotiated channel\\
$\eta$ & 0.5 & adjustment speed & assumed\\
own-capital share & 0.15 & Bauer (2026d), \S5.1 & inherited\\
$\kbar$ for counts & 0.60 & Bauer (2026d), regime counts & inherited\\
\bottomrule
\end{tabular}
\end{table}

\paragraph{Notation.}
This paper adds three symbols to the model: $D$ (doctrine map), $\lbar$ (maximal doctrinal
uplift) and $c$ (the mean of the case-level threshold, which locates the case population).
The scale $\varsigma$ of the logistic law appears only in the benchmark of
Theorem~\ref{thm:maxwell}(iv) and in the logistic row of Table~\ref{tab:rob}; it is not a
parameter of the model. Three collisions were live and
are avoided as follows. $\Phi$ is \emph{not} used for the map, because $\varphi$ is the
capability-building parameter in the companion papers' law of motion; hence $D$. The symbol $c$
collides with $c_0$, the enforcement-cost floor of the companion papers, which is a different
object; $c$ here always carries no subscript and always denotes the location of
$m_{\mathrm{crit}}$. And $\varsigma$ is used rather than $\sigma$, which is the noise intensity of
Section~\ref{sec:stoch}. The cross-paper registry
(doi:10.6084/m9.figshare.33193764) is the authoritative record, and its count is not duplicated
here; among the symbols it carries as multiply assigned are $\xi$, $\psi$, $\varphi$, $\kappa$,
$\delta$, $b$, $m$, $c_0$ and $\theta$. The symbols added here are appended to it.

\paragraph{Notation across the series.}
Three symbols are chosen to avoid collisions with companion papers. The excess function is
$G(m) = D(m) - m$: $g$ is the capability index of the companion paper, where it carries a
falsification condition and the replication deposit, and is the more expensive of the two to
move. The population of drafting parties in Section~\ref{sec:stoch} is $n$, which leaves $N$ to
the case flow $N_t$ of Section~\ref{sec:falsify} and to the triple $(m, F, N)$ named there. The
rate at which that flow thins is $\Gamma$, because $\nu$ is the gain in base provability from
better forensics in the companion paper and $\vartheta$, a variant glyph of $\theta$, would sit
next to the most heavily used symbol of the series, which carries verifiability in four of its
seven papers. $\Gamma$ is the conventional label for an escape rate in the Freidlin--Wentzell setting
of Section~\ref{sec:stoch}; elsewhere in the series it occurs once, as a deviation payoff inside one
proof of \citet{bauer2026b}, and it never meets the outflow rate in one expression.

\paragraph{Symbols with several readings.}
Five symbols carry more than one reading across the series and are recorded here rather than
renamed. $\eta$ is the adjustment speed in this paper, the distinguishable range in the first
paper of the series --- where it is central, and that paper is published --- and a
labour-sorting friction in two others. $\varsigma$ is the scale of the logistic benchmark here
--- not a parameter of the model --- and a certificate signal inside one proof of the
first paper. $\lbar$ is the maximal
doctrinal uplift here; the insurance paper of the series writes two loadings with the same
letter, both subscripted, which separates them typographically. $S$ is
the logistic law of that benchmark here and, where it denotes penalty amounts elsewhere, always
subscripted. $D$
is the doctrine map here, the liability pledge in the first and fifth papers, and a regime
indicator in the forthcoming empirical paper; the three never meet in one expression, and the
empirical paper, still in preparation, is the one that will clear the letter.

\paragraph{Multiplicity inside this paper.}
Two symbols carry more than one reading here and are flagged for the reader rather than renamed.
$\pi$ is the AI failure rate of Section~\ref{sec:map}, the invariant density $\pi_{\sigma}$ of
Proposition~\ref{prop:potential} --- written bare in its proof --- and the circle constant in
the Kramers prefactor of Proposition~\ref{prop:mfpt}. $T$ is the policy horizon in
$\sigma^{\dagger}(T)$, the adjustment map $T(m)$ of Proposition~\ref{prop:stab} together with
its solvency-constrained variant, and the expected first-passage function $T(x) = \E_x[\tau]$ in
Appendix~\ref{app:proofs}.

\section{Data availability}
\label{app:data}

Every number, table entry and figure in this paper is produced by two self-contained Python
modules released with the manuscript. Both use the case population of
Section~\ref{sec:calib} in closed form, implemented in \texttt{population.py}, so that no random
draw enters any reported number. \texttt{make\_doctrine.py} carries eight labelled blocks: the
calibration, the fixed points with the multiplicity threshold $\lbar^{\dagger}$, the folds and
their continuation, the closed form of Proposition~\ref{prop:collapse} with a numerical
verification against direct iteration (maximum deviation $0.0$) and the duration $t^{*}$ of
Corollary~\ref{cor:irrev}, the shock experiment, the jurisdictional taxonomy, the robustness
grid, and the regime-map replication. No proprietary or licensed data are used: all inputs are
either published parameter values cited in the text or synthetic values generated by the code, so
the package verifies and illustrates the derivations and is not evidence for them.

The regime-map block reproduces the companion paper's corollary ``The exclusions move the regime boundary''
exactly under exogenous $m$
(three and eleven solvency-bound cells; survival counts eighteen and twelve at $\xi = 0.2$
falling to twelve and four at $\xi = 0.4$), which we report as the package's principal external
check.

The second module, \texttt{make\_stochastic.py}, produces every number and figure of
Section~\ref{sec:stoch} in nine labelled blocks: the potential, the Maxwell point with its
closed-form check, the exact first-passage quadrature with Kramers asymptotics and the
Euler--Maruyama validation (the only simulation in the package, seed 20260811), the noise
thresholds $\sigma^{\dagger}$ and $\sigma^{\ddagger}$, the policy-horizon threshold
$c^{\dagger}(\sigma, T)$, the form-free demonstrations (normal, logistic and bimodal maps), the
monoculture bridge with the level-versus-spread comparison, the provability channel, and the
case $\lbar = 4$; results are serialised to \texttt{stochastic\_results.json}. A consistency
gate, \texttt{check\_doctrine.py}, recomputes the headline numbers from the serialised results
and the population, checks every number printed in the manuscript against them, and refuses to
pass on any mismatch. The Lean sources of Section~\ref{sec:lean} are shipped in
\texttt{lean/DoctrineOrder.lean} together with the compiler transcript. The package --- manuscript source, figures, generator modules, consistency gate,
serialised results, and the Lean sources with transcript --- is archived at
\texttt{doi:10.6084/m9.figshare.33212916}.

An interactive companion, the Penalty Ceiling Explorer, recomputes the doctrine map and its fixed
points, the folds, the Maxwell point, the escape times and the fuse in the browser from the same
case population: \url{https://tafew.github.io/doctrine-fixed-point/}. Its source and a
verification script, which compares 44 of its values with the serialised results and reads its
opening state against \texttt{population.py}, are at
\url{https://github.com/Tafew/doctrine-fixed-point} and in the package under
\texttt{explorer/}. The explorer is a companion and not the source of record; the package is.

\section*{Declarations}

\paragraph{Generative AI and AI-assisted technologies in the writing process.}
During the preparation of this work the author used Anthropic's Claude, a large language model,
for literature search and source verification, drafting and editing of text, derivation checking,
and the production of the replication code and figures. The author reviewed and edited all
content and takes full responsibility for the content of this article.

\paragraph{Competing interests.} The author is Managing Director and Founder of Aegis Compliance
and Strategies O\"U, a compliance and strategy advisory firm, and is author and publisher of
\emph{Diebstahlsicher} (Bauer Advanced Network Solutions KG, Vienna, September 2026, ISBN
978-3-9506352-0-1), a German-language practitioner book that draws on this line of work.

\paragraph{Funding.} This research received no external funding.

\paragraph{CRediT authorship contribution.} Andreas Bauer: Conceptualization, Methodology, Formal
analysis, Software, Investigation, Writing --- original draft, Writing --- review \& editing.

\end{document}